\documentclass[11pt]{article}
\usepackage[pdftex]{graphicx,color}
\usepackage{graphicx,color}
\usepackage{amsmath,amssymb,dsfont,enumerate,fullpage,epsfig,multirow,longtable}
\usepackage{epsfig,epstopdf}
\usepackage{cases}
\usepackage{algorithm, algorithmic}

\usepackage[numbers]{natbib}

\usepackage{authblk}

\newenvironment{proof}{\noindent\emph{Proof\ }}{\hspace*{\fill}$\Box$\medskip}
\newenvironment{claimproof}{\noindent\emph{Proof of claim\ }}{\hspace*{\fill}$\Box$\medskip}

\newtheorem{theorem}{Theorem}
\newtheorem{lemma}{Lemma}
\newtheorem{claim}{Claim}

\newtheorem{corollary}{Corollary}
\usepackage{amsmath}
\usepackage{paralist}
\usepackage{framed}

\newcommand\restr[2]{{
  \left.\kern-\nulldelimiterspace 
  #1 
  \vphantom{\big|} 
  \right|_{#2} 
  }}

\newcommand{\one}{\ensuremath{\mathds{1}}}

\title{Lagrangian Duality based Algorithms in Online Scheduling}
\author{Nguyen Kim Thang\thanks{Research supported by FMJH program Gaspard Monge in Optimization and
Operations Research and by EDF.}}
\affil{IBISC, University of Evry Val d'Essonne, France 
}

\date{}

\begin{document}

\maketitle

\begin{abstract}
We consider Lagrangian duality based approaches to design and analyze 
algorithms for online energy-efficient scheduling. 
First, we present a primal-dual framework. 
Our approach makes use of the Lagrangian weak duality and convexity to derive 
dual programs for problems which could be formulated as convex assignment
problems. The duals have intuitive structures as the ones in linear programming.
The constraints of the duals explicitly indicate the online decisions
and naturally lead to competitive algorithms. 
Second, we use a dual-fitting approach, which also based on the weak duality, 
to study problems which are unlikely to admit convex relaxations. Through the analysis,
we show an interesting feature in which 
primal-dual gives idea for designing algorithms while the analysis 
is done by dual-fitting.

We illustrate the advantages and the flexibility of the approaches through problems 
in different setting: from single machine 
to unrelated machine environments, from typical competitive analysis to the one with
resource augmentation, from convex relaxations to non-convex relaxations.
\end{abstract}

\newpage 
\tableofcontents

\section{Introduction}	\label{sec:intro}

In the online setting, items arrive over time 
and one must determine how to serve items in order to optimize a quality of 
service without the knowledge about future. 
A popular measure for studying the performance of online algorithms is 
\emph{competitive ratio} in the model of the worst-case analysis. 
An algorithm is said to be $c$-\emph{competitive}
if for any instance its objective value is within factor $c$ of the optimal offline algorithm's one.  
Moreover, to remedy the limitation of pathological instances in the worst-case analysis, 
there is other model called \emph{resource augmentation} \cite{KalyanasundaramPruhs00:Speed-is-as-powerful}. 
In the latter, online algorithms are given an extra power and are compared to the optimal 
offline algorithm without that additional resource. This model has successfully provided theoretical evidence for 
heuristics with good performance in practice, especially in online scheduling where 
jobs arrive online and need to be processed on machines. 
We say a scheduling algorithm is \emph{$s$-speed $c$-competitive}
if for any input instance the objective value of the algorithm with machines of 
speed $s$ is at most $c$ times the objective value of the optimal offline scheduler
with unit speed machines. 

The most successful tool until now to analyze online algorithms 
is the potential function method. 
Potential functions have been designed to show that the corresponding algorithms behave
well in an amortized sense. However, designing such potential functions is far from trivial and often
yields little insight about how to design such potential functions and algorithms for related problems. 

Recently, interesting approaches 
\cite{AnandGarg12:Resource-augmentation,GuptaKrishnaswamy12:Online-Primal-Dual,Thang13:Lagrangian-Duality} 
based on mathematical programming have been presented 
in the search for a principled tool to design and analyze online scheduling algorithms. 
The approaches give insight about the nature of many scheduling problems,
hence lead to algorithms which are usually simple and competitive.

\subsection{Approaches and Contribution}	
In this paper, we present approaches based on Lagrangian duality in designing and analyzing 
algorithms for online energy-efficient scheduling problems. 

\paragraph{Primal-Dual Approach.} We first show a primal-dual framework for a general online 
convex assignment problem and its applications to online scheduling. In the framework, the algorithm 
decisions are interactively guided by the dual variables in the primal-dual sense. 

The online convex assignment problem consists of a set of agents and a set of items which arrive over time.
At the arrival of item $j$, the item needs to be (fractionally) assigned to some agents $i$. 
Let $x_{ij}$ is the amount of $j$ assigned to $i$.
The problem is to minimize $\sum_{i} f_{i} (\sum_{j} a_{ij}x_{ij} )$  
under some constraints $g_{i} (\sum_{j} b_{ij} x_{ij}) \leq 0$
and $h_{j}(\sum_{i} c_{ij} x_{ij}) \leq 0$ for every $i,j$ where functions $f_{i},g_{i},h_{j}$'s are convex.
In offline setting, the optimal solutions are completely characterized by 
the KKT conditions (see \cite{BoydVandenberghe04:Convex-Optimization} for example).
However, for online setting, the conditions could not be satisfied due to the lack of knowledge
about the future. 

Our approach is the following. We first consider the problem as a primal convex mathematical program. 
Then we derive a Lagrangian dual program by the standard Lagrangian duality. 
Instead of analyzing directly the corresponding Lagrangian functions
where in general one can not disentangle the objective and the constraints
as well as the primal and dual variables,
we exploit the convexity property of given functions and construct 
a dual program. In the latter, dual variables are separated from the primal ones. 
The construction is shown in Section \ref{sec:framework}.
As the price of the separation procedure, the strong duality 
property is not guaranteed. However, the weak duality always holds and that is 
crucial (and sufficient) to deduce a lower bound for the given problem.
The dual construction is not standard in optimization but the goal is to derive 
duals with intuitive structures similar to the ones in linear programming which 
are easier to work with.
An advantage of the approach lies in the dual program
in which the constraints could be maintained online. Moreover, the dual constraints 
explicitly indicate the online decisions and naturally lead to a competitive algorithm
in the primal-dual sense.   

The dual construction is inspired by \cite{DevanurJain12:Online-matching}
which used the primal-dual approach for online
matching with concave return. In fact, Devanur and Jain \cite{DevanurJain12:Online-matching}
considered a matching problem with convex objective function and \emph{linear} constraints.
They linearized the objective function and derived the dual in the sense of linear programming. 
In our framework, we consider problems with convex objective and 
\emph{convex} constraints\footnote{The primal-dual machinery of linear programming cannot be applied anymore.}. 
We then construct our duals from Lagrangian dual programs. 
Informally, the construction could be seen as a linearization of Lagrangian duals. 

\paragraph{Applications.}	 We illustrate the advantages and the flexibility of the approach through 
online scheduling problems related to throughput. In the setting, there are a set of 
unrelated machines. Each job $j$ is released at time $r_{j}$, has deadline $d_{j}$, a value $a_{j}$ and 
a processing volume $p_{ij}$ if job $j$ is executed in machine $i$. 
Jobs could be executed preemptively but migration is not allowed, 
i.e., no job could be executed in more than one machine. At any time $t$, the scheduler 
has to choose an assignment of jobs to machines and the speed of each machine 
in order to process such jobs. The energy cost of machine $i$ is 
$\int_{0}^{\infty} P(s_{i}(t))dt$ where $P$ is a given convex energy power and $s_{i}(t)$
is the speed of machine $i$ at time $t$. Typically, $P(z) = z^{\alpha}$ for some constant $\alpha \geq 1$. 
In the setting, we look for competitive and energy-efficient algorithms. The following objectives 
are natural ones representing the tradeoff between value and energy. 
The first objective is to minimize energy cost plus the \emph{lost value} 
--- which is the total value of uncompleted jobs. The second objective is to maximize 
 the total value of completed jobs minus the energy cost.

\begin{enumerate}
	\item For the objective of minimizing energy plus the lost value 
	we derive a primal-dual algorithm for the single machine setting. 
	The competitive ratio is characterized by a system
	of differential equations. For a specific case where $P(z) = z^{\alpha}$, the competitive
	ratio turns out to be $\alpha^{\alpha}$ (and recognize the result in \cite{KlingPietrzyk13:Profitable-Scheduling}. 
	With the primal-dual framework, the result is more general and the analysis is simpler.
	\item For the objective of maximizing the total value of completed jobs minus the energy cost,
	it has been shown that without resource augmentation 
	no algorithm has bounded competitive ratio even for a single machine \cite{PruhsStein10:How-to-Schedule-When}.
	We study the problem for unrelated machines in the resource augmentation model. We give a primal-dual 
	algorithm which is $(1+\epsilon)$-speed and $1/\epsilon$-competitive for 
	every $\epsilon \geq \epsilon(P) > 0$ where $\epsilon(P)$ depends on function $P$.
	For typical function $P(z) = z^{\alpha}$, $\epsilon(P) = 1 - \alpha^{-1/\alpha}$
	which is closed to 0 for large $\alpha$.
\end{enumerate}

Note that for these problems, we consider relaxations with convex objectives and linear constraints.

\paragraph{Dual-fitting approach.} 
An essential point of the primal-dual approach is a \emph{convex} relaxation 
of the corresponding problems. However, some problems unlikely admit such a relaxation. 
To overcome that difficulty, we follow the dual-fitting approach for \emph{non-convex}
programming presented in \cite{Thang13:Lagrangian-Duality}. A summary of the approach is as 
follows.

Given a problem, formulate a relaxation which is \emph{not} necessarily convex and its Lagrangian dual. 
Next construct dual variables such that the Lagrangian dual has objective value within a desired factor 
of the primal one (due to some algorithm). Then by the standard Lagrangian weak 
duality\footnote{For completeness, the proof of weak duality is given in the appendix} 
for mathematical programming, the competitive ratio follows. Since the Lagrangian weak 
duality also holds in the context of calculus of variations, the approach could be applied 
for the unknowns which are not only variables but also functions.  

Let $L(x,\lambda)$ be the Lagrangian function with primal and dual variables $x$
and $\lambda$, respectively. 
Let $\mathcal{X}$ and $\mathcal{Y}$ are feasible sets of $x$ and $\lambda$.
Intuitively, the approach could be seen as a game between an algorithm
and an adversary. The algorithm chooses dual variables $\lambda^{*}$ in such a way that whatever 
the choice (strategy) of the adversary, the value $\min_{x \in \mathcal{X}} L(x,\lambda^{*})$
is always within a desirable factor $c$ of the objective due to the algorithm.
We emphasize that $\min_{x \in \mathcal{X}} L(x,\lambda^{*})$ is taken over $x$ feasible solutions 
of the primal.

An advantage of the approach is the flexibility of the formulation. 
As convexity is not required, we can come up with a direct and natural relaxation for the problem.
The main core of the approach is to determine the dual variables
and to prove the desired competitive ratio. Determining such dual variables is the crucial 
step in the analysis. However, the dual variables usually have intuitive interpretations which are useful  
to figure out appropriate values of such variables. Besides, the dual variables are not interactively 
constructed as in the primal-dual approach --- this is the main difference between two approaches.
Nevertheless, for some problems one could informally separate the convex and non-convex parts. Then 
the dual solution for the original problem may be derived from 
a dual solution for the convex part (constructed using primal-dual) by
adding some correcting terms due to the non-convex part. 


\paragraph{Applications.}
We consider the general energy model: speed scaling with power down. 
There is a machine which can be set either in the sleep state or in the active 
state. Each transition of the machine from the sleep state to the active one has 
cost $A$, which represents the \emph{wake-up} cost. 
In the sleep state, the energy consumption of the machine is 0. 
The machine, in its active state, can choose a speed $s(t)$
to execute jobs. The power energy consumption of 
the machine at time $t$ in its active state is $P(s(t)) = s(t)^{\alpha} + g$
where $\alpha \geq 1$ and $g \geq 0$ are characteristic parameters of 
the machine. Hence, the consumed energy (without wake-up cost) of the machine
is $\int_{0}^{\infty} P(s(t))dt$ where the integral is taken during the machine's 
active periods. We decompose the latter into \emph{dynamic energy} 
$\int_{0}^{\infty} s^{\alpha}(t)dt$ and \emph{static energy} 
$\int_{0}^{\infty}g dt$ (where again the integrals are taken during
active periods).
Jobs arrive over time, a job $j$ is released at time $r_{j}$, has weight $w_{j}$ 
and requires $p_{j}$ units of processing volume if it is processed 
on machine $i$. A job could be processed preemptively.
At any time, the scheduler has to determine the state and the speed of every machine (it it is active) 
and also a policy to execute jobs.
We consider two problems in the setting.

In the first problem, each job $j$ has additionally a deadline $d_{j}$ by which 
the job has to be completed. The objective is to minimize the total consumed
energy.  

In the second problem, jobs do not have deadline. Let $C_{j}$ 
be the completion time of the job $j$. 
The \emph{flow-time} of a job $j$ is defined as $C_{j} - r_{j}$, 
which represented the waiting time of $j$ on the server.
The objective is to minimize the total weighted flow-time of all jobs plus the total 
energy. 

As posed in \cite{Albers10:Energy-efficient-algorithms}, an important direction in energy-efficient scheduling
is to design competitive algorithms for online problems
in the general model of speed scaling with power down. 
Attempting efficient algorithms in the general energy model, one encounters 
the limits of current tools which have been successfully applied 
in previous energy models. That results in a few work on the model
\cite{AlbersAntoniadis12:Race-to-idle:,BampisDurr12:Speed-scaling,HanLam10:Deadline-scheduling},
in contrast to the widely-studied models of speed scaling only or power down only. 
The potential function method, as mentioned earlier, yield little insight
on the construction of new algorithms in this general setting.     
Besides, different proposed approaches based on the duality of 
mathematical programming
\cite{AnandGarg12:Resource-augmentation,GuptaKrishnaswamy12:Online-Primal-Dual,DevanurHuang14:Primal-Dual} 
require that the problems admit linear of convex relaxations. However, it is unlikely to 
formulate problems in the general energy model as convex 
programs.  

Our results in the general energy model are the following.

\begin{enumerate}
	\item For the problem of minimizing the total consumed energy,
	we formulate a natural \emph{non-convex} relaxation using the Dirac delta function. 
	We first revisit a special case with no wake-up cost under the primal-dual view. In this case,
	the relaxation becomes convex and our framework could be applied to show a
	$\alpha^{\alpha}$-competitive algorithm (the algorithm is in fact algorithm 
	\textsc{Optimal Available} \cite{YaoDemers95:A-Scheduling-Model}).
	Next we study the general problem with wake-up cost. 
	The special case effectively gives ideas to determine the machine speed in active state. 
	Thus we consider an algorithm in which the procedure maintaining the machine 
	speed in active state follows the ideas in the special case. The algorithm turns out to be algorithm 
	\textsc{Sleep-aware Optimal Available} (SOA) \cite{HanLam10:Deadline-scheduling} with different description 
	(due to the primal-dual view). \citet{HanLam10:Deadline-scheduling} proved that SOA has competitive ratio 
	$\max\{4,\alpha^{\alpha}+2\}$. We prove that SOA is indeed  
	$\max\{4,\alpha^{\alpha}\}$-competitive by the dual-fitting technique.
	Although the improvement is slight, the analysis is 
	\emph{tight}\footnote{The algorithm has competitive ratio exactly $\alpha^{\alpha}$
	even without wake-up cost \cite{BansalKimbrel07:Speed-scaling}.}
	and it suggests that the duality-based approach is seemingly a right tool for online scheduling. 
	Through the problem, we illustrate an interesting feature in the construction of  
	algorithms for non-convex relaxations. The primal-dual framework
	gives ideas for the design of an algorithm while the analysis is done using dual-fitting technique.     
	\item For the problem of minimizing energy plus weighted flow-time, we derive a $O(\alpha/\log \alpha)$-competitive
	algorithm using the dual fitting framework; that matches the best known competitive ratio (up to a constant)
	for the same problem in the restricted speed scaling model 
	(where the wake-up cost and the static energy cost are 0). 
	Informally, the dual solutions are constructed as the combination of a solution 
	for the convex part of the problem and a term that represents the lost due to 
	the non-convex part. Building upon the salient ideas of the previous analysis,
	we manage to show the competitiveness of the algorithm. 
\end{enumerate}

\subsection{Related work}


In the search for principled methods to design and analyze online problems,
especially in online scheduling, interesting approaches 
\cite{AnandGarg12:Resource-augmentation,GuptaKrishnaswamy12:Online-Primal-Dual,Thang13:Lagrangian-Duality} 
based on mathematical programming have been presented. 
The approaches give insight about the nature of many scheduling problems,
hence lead to algorithms which are usually simple and competitive
\cite{AnandGarg12:Resource-augmentation,GuptaKrishnaswamy12:Online-Primal-Dual,Thang13:Lagrangian-Duality,DevanurHuang14:Primal-Dual,ImKulkarni14:Competitive-Algorithms,ImKulkarni14:SELFISHMIGRATE:-A-Scalable}.

\citet{AnandGarg12:Resource-augmentation} was the first who proposed 
studying online scheduling by linear (convex) 
programming and dual fitting. By this approach, they gave 
simple algorithms and simple analyses with improved performance
for problems where the analyses based on potential functions are complex or 
it is unclear how to design such functions. Subsequently, 
Nguyen \cite{Thang13:Lagrangian-Duality} generalized the approach
in \cite{AnandGarg12:Resource-augmentation} and proposed to study online scheduling 
by non-convex programming and the weak Lagrangian duality. Using that technique,
\cite{Thang13:Lagrangian-Duality} derive competitive algorithms 
for problems related to weighted flow-time.

\citet{BuchbinderNaor09:The-Design-of-Competitive}
presented the primal-dual method for online packing and covering problems.
Their method unifies several previous
potential function based analyses and is a powerful tool to design and analyze 
algorithms for problems with linear relaxations.
\citet{GuptaKrishnaswamy12:Online-Primal-Dual}
gave a primal-dual algorithm for a general class of scheduling problems with 
cost function $f(z) = z^{\alpha}$. \citet{DevanurJain12:Online-matching}
also used the primal-dual approach to derive optimal competitive ratios for online
matching with concave return. 
The construction of dual programs in \cite{DevanurHuang14:Primal-Dual,DevanurJain12:Online-matching}
is based on convex conjugates and Fenchel duality for primal convex programs
in which the objective is convex and the constraints are \emph{linear}. 


An interesting quality of service in online scheduling is the tradeoff between 
energy and throughput. The online problem to minimize the consumed energy plus lost values 
with the energy power $P(z) = z^{\alpha}$ is first studied by \cite{ChanLam10:Tradeoff-between}
where a $(\alpha^{\alpha} + 2 e \alpha)$-competitive algorithm is given for a single machine.
Subsequently, \citet{KlingPietrzyk13:Profitable-Scheduling} derived an 
improved $\alpha^{\alpha}$-competitive for identical machines with migration using the technique in 
\cite{GuptaKrishnaswamy12:Online-Primal-Dual}. 
The online problem to maximize the total value of completed jobs minus the consumed energy
for a single machine has been considered in
\cite{PruhsStein10:How-to-Schedule-When}. \citet{PruhsStein10:How-to-Schedule-When} proved that 
the competitive ratio without resource augmentation is unbounded and gave 
an $(1+\epsilon)$-speed, $O(1/\epsilon^{3})$-competitive algorithm for a single machine.
 
The objective of minimizing the total flow-time plus energy has been widely studied in speed scaling energy model. 
For a single machine, \citet{BansalChan09:Speed-scaling}
gave a $(3+\epsilon)$-competitive algorithm. Besides, they also proved a  
$(2+\epsilon)$-competitive algorithm for minimizing total \emph{fractional} weighted flow-time plus
energy. Their results hold for a general class of convex power functions. Those results also 
imply an $O(\alpha/\log \alpha)$-competitive algorithm for weighted flow-time plus energy 
when the energy function is $s^{\alpha}$. Again, always based on linear programming and dual-fitting, 
\citet{AnandGarg12:Resource-augmentation} proved an $O(\alpha^{2})$-competitive algorithm 
for unrelated machines. Subsequently, Nguyen \cite{Thang13:Lagrangian-Duality} and
\citet{DevanurHuang14:Primal-Dual} presented an $O(\alpha/\log \alpha)$-competitive algorithms 
for unrelated machines by dual fitting and primal dual approaches, respectively. It turns out that the different
approaches lead to the same algorithm. To the best of our knowledge, this objective is not studied 
in the speed scaling with power down energy model.

In the speed scaling with power down energy model, all previous papers
considered the problem of minimizing the energy consumption on a single machine. 
\citet{IraniShukla07:Algorithms-for-power} was the first 
who studied the problem in online setting and derived an algorithm with competitive ratio 
$(2^{2\alpha-2}\alpha^{\alpha} + 2^{\alpha-1} + 2)$. Subsequently, 
\citet{HanLam10:Deadline-scheduling} presented an algorithm 
which is $\max\{4,\alpha^{\alpha}+2\}$-competitive.
In offline setting, the problem is recently showed to be NP-hard 
\cite{AlbersAntoniadis12:Race-to-idle:}. Moreover, \citet{AlbersAntoniadis12:Race-to-idle:}
also gave a 1.171-approximation algorithm, which improved
the 2-approximation algorithm in \cite{IraniShukla07:Algorithms-for-power}. 
If the instances are agreeable then the problem is polynomial 
\cite{BampisDurr12:Speed-scaling}.

\subsection{Organization}
The paper is organized as follows. In Section \ref{sec:framework}, we introduce the 
online convex assignment problem and present a primal-dual framework for this problem.
In Section \ref{sec:energy+values} and Section \ref{sec:values-energy}, we apply the framework
to derive primal-dual algorithms for problems related to the tradeoff between energy and value. 
In Section \ref{sec:4S-energy} and Section \ref{sec:4S-energy+flow}, we study problems
in the speed scaling with power down model using the dual-fitting approach.  
In the former, we study the problem of minimizing energy and 
in the latter we consider the problem of minimizing 
the total energy plus weighted flow-time. In the beginning of each section, we 
restate the considered problem in a short description.

\section{Framework for Online Convex Assignment}		\label{sec:framework}
  
Consider the assignment problem where items $j$ arrive online 
and need to be (fractionally) assigned to some agents $i$ with the following objective and constraints.   
  \begin{alignat}{3}	\label{prog:primal}
    \text{min} \quad P(x) :=  \sum_{i} f_{i}  &\biggl( \sum_{j} a_{ij}x_{ij} \biggl)  \tag{$\mathcal{P}$} \\
	\text{subject to} \quad   g_{i}\biggl(\sum_{j} b_{ij} x_{ij}\biggl) &\leq 0  \qquad \forall i	\notag \\
					     h_{j}\biggl(\sum_{i} c_{ij} x_{ij}\biggl) &\leq 0 \qquad \forall j 	\notag \\ 
	    				     x_{ij} &\geq 0 \qquad \forall i,j	\notag 
  \end{alignat}
where $x_{ij}$ indicates the amount of item $j$ assigned to agent $i$ 
and functions $f_{i},g_{i},h_{j}$ are convex, differential for every $i,j$ and 
$a_{ij}, b_{ij} \geq 0$. Denote $k \prec j$ if item $k$ is released before item $j$. 

Let $\mathcal{X}$ be the set of feasible solutions of $(\mathcal{P})$. 
The Lagrangian dual is
$
\max_{\lambda,\gamma \geq 0} \min_{x \in \mathcal{F}} L(x,\lambda,\gamma)
$ 
where $L$ is the following Lagrangian function 
\begin{align*}
L(x,\lambda,\gamma) &= \sum_{i} f_{i}\biggl( \sum_{j} a_{ij}x_{ij} \biggl) 
+ \sum_{i} \lambda_{i} g_{i}\biggl(\sum_{j} b_{ij} x_{ij}\biggl)
+ \sum_{j} \gamma_{j} h_{j}\biggl(\sum_{i} c_{ij} x_{ij}\biggl) \notag \\
&\geq \sum_{i,j} (x_{ij} - x^{*}_{ij})\biggl[ a_{ij}f'_{i}\biggl( \sum_{k} a_{ik}x^{*}_{ik} \biggl) 
	+ \lambda_{i} b_{ij}g'_{i}\biggl(\sum_{k} b_{ik} x^{*}_{ik}\biggl)  
	+ \gamma_{j}c_{ij} h'_{j}\biggl(\sum_{i'} c_{i'j} x^{*}_{i'j}\biggl) \biggl] \notag \\
& \qquad +  \sum_{i} f_{i}\biggl( \sum_{j} a_{ij}x^{*}_{ij} \biggl)
	+ \sum_{i} \lambda_{i} g_{i}\biggl(\sum_{j} b_{ij} x^{*}_{ij}\biggl)
	+ \sum_{j} \gamma_{j} h_{j}\biggl(\sum_{i} c_{ij} x^{*}_{ij}\biggl)  \\
&\geq \sum_{i,j} (x_{ij} - x^{*}_{ij})\biggl[ a_{ij}f'_{i}\biggl( \sum_{k \prec j} a_{ik}x^{*}_{ik} \biggl) 
	+ \lambda_{i} b_{ij}g'_{i}\biggl(\sum_{k \prec j} b_{ik} x^{*}_{ik}\biggl)  
	+ \gamma_{j}c_{ij} h'_{j}\biggl(\sum_{i'} c_{i'j} x^{*}_{i'j}\biggl) \biggl] \notag \\
& \qquad +  \sum_{i} f_{i}\biggl( \sum_{j} a_{ij}x^{*}_{ij} \biggl)
	+ \sum_{i} \lambda_{i} g_{i}\biggl(\sum_{j} b_{ij} x^{*}_{ij}\biggl)
	+ \sum_{j} \gamma_{j} h_{j}\biggl(\sum_{i} c_{ij} x^{*}_{ij}\biggl)  
\end{align*}
where the inequalities holds for any $x^{*}$ due the convexity of functions $f_{i},g_{i},h_{j}$'s.
In the first inequality, we use $f_{i}(z) \geq f_{i}(z^{*}) + (z-z^{*})f'_{i}(z^{*})$ 
(similarly for functions $g_{i},h_{j}$'s) and in the second inequality, 
we use the monotonicity of $f'_{i}$ (and similarly for $g'_{i}$).
Denote 
\begin{align*}
 M(x,x^{*},\lambda,\gamma) &:= \sum_{i,j} (x_{ij} - x^{*}_{ij}) \biggl[ a_{ij}f'_{i}\biggl( \sum_{k \prec j} a_{ik}x^{*}_{ik} \biggl) 
	+ \lambda_{i} b_{ij}g'_{i}\biggl(\sum_{k \prec j} b_{ik} x^{*}_{ik}\biggl)  
	+ \gamma_{j}c_{ij} h'_{j}\biggl(\sum_{i} c_{ij} x^{*}_{ij}\biggl) \biggl] \\
 N(x^{*},\lambda, \gamma) &:= \sum_{i} f_{i}\biggl( \sum_{j} a_{ij}x^{*}_{ij} \biggl)
	+ \sum_{i} \lambda_{i} g_{i}\biggl(\sum_{j} b_{ij} x^{*}_{ij}\biggl)
	+ \sum_{j} \gamma_{j} h_{j}\biggl(\sum_{i} c_{ij} x^{*}_{ij}\biggl)  
\end{align*}
We have 
\begin{align} 	\label{framework-Lagrangian}
L(x,\lambda,\gamma) \geq M(x,x^{*},\lambda,\gamma) + N(x^{*},\lambda,\gamma)
\end{align}
Intuitively, one could imagine that $x^{*}$ is the solution of an algorithm (or a function 
on the solution of an algorithm). We emphasize that $x^{*}$ is \emph{not} a solution of 
an optimal assignment. The goal is to design 
an algorithm, which produces $x^{*}$ and derives dual variables $\lambda, \gamma$, in such
a way that the primal objective is bounded by a desired factor from the dual one.
   
Inequality (\ref{framework-Lagrangian}) naturally leads to the following idea of an algorithm. 
For any item $j$, we maintain the following invariants  
\begin{align*}
a_{ij}f'_{i}\biggl( \sum_{k \prec j} a_{ik}x^{*}_{ik} \biggl) 
	+ \lambda_{i} b_{ij}g'_{i}\biggl(\sum_{k \prec j} b_{ik} x^{*}_{ik}\biggl)  
	+ \gamma_{j}c_{ij} h'_{j}\biggl(\sum_{i} c_{ij} x^{*}_{ij}\biggl)
	&\geq 0 ~\forall i \\
a_{ij}f'_{i}\biggl( \sum_{k \prec j} a_{ik}x^{*}_{ik} \biggl) 
	+ \lambda_{i} b_{ij}g'_{i}\biggl(\sum_{k \prec j} b_{ik} x^{*}_{ik}\biggl)  
	+ \gamma_{j}c_{ij} h'_{j}\biggl(\sum_{i} c_{ij} x^{*}_{ij}\biggl)
	&= 0 ~\text{if } x^{*}_{ij} > 0 
\end{align*}
Whenever the invariants hold for every $j$,
$M(x,x^{*},\lambda,\gamma) \geq 0$ since $x_{ij} \geq 0$ for every $i,j$. 
Therefore, $L(x,\lambda,\gamma) \geq N(x^{*},\lambda,\gamma)$ and so the dual 
is lower-bounded by $N(x^{*},\lambda,\gamma)$, which does not depend anymore
on $x$. The procedure of maintaining the invariants dictate the decision $x^{*}$ 
of an algorithm and indicates the choice of dual variables. 

Consider the following dual
  \begin{alignat}{3}	\label{prog:dual}
    \text{max} \quad  N(x^{*},\lambda,\gamma)  \tag{$\mathcal{D}$} \\
	\text{subject to} \quad   a_{ij}f'_{i}\biggl( \sum_{k \prec j} a_{ik}x^{*}_{ik} \biggl) 
	+ \lambda_{i} b_{ij}g'_{i}\biggl(\sum_{k \prec j} b_{ik} x^{*}_{ik}\biggl)  
	+ \gamma_{j}c_{ij} h'_{j}\biggl(\sum_{i} c_{ij} x^{*}_{ij}\biggl)
		&\geq 0 \qquad \forall i,j 	\notag \\
	a_{ij}f'_{i}\biggl( \sum_{k \prec j} a_{ik}x^{*}_{ik} \biggl) 
	+ \lambda_{i} b_{ij}g'_{i}\biggl(\sum_{k \prec j} b_{ik} x^{*}_{ik}\biggl)  
	+ \gamma_{j}c_{ij} h'_{j}\biggl(\sum_{i} c_{ij} x^{*}_{ij}\biggl) 
	&= 0 ~\text{if } x^{*}_{ij} > 0 \qquad \forall i,j 	\notag \\ 
	    	  x^{*},\lambda,\gamma &\geq 0 \qquad \forall i,j	\notag
  \end{alignat}

\begin{lemma}[Weak Duality]
Let $OPT(\mathcal{P})$ and $OPT(\mathcal{D})$ be optimal values of primal program
$(\mathcal{P})$ and dual program $(\mathcal{D})$, respectively. Then
$OPT(\mathcal{P}) \geq OPT(\mathcal{D})$.
\end{lemma}
\begin{proof}
It holds that
$$
OPT(\mathcal{P}) \geq \max_{\lambda,\gamma \geq 0} \min_{x \in \mathcal{X}} L(x,\lambda,\gamma)
\geq N(x^{*},\lambda,\gamma)
$$
where the inequalities follow the weak Lagrangian duality and the constraints of $(\mathcal{D})$
for every feasible solution $x^{*},\lambda,\gamma$.
Therefore, the lemma follows.
\end{proof}

Hence, our framework consists of maintaining the invariants for every online item $j$ and
among feasible set of dual variables (constrained by the invariants) choose the ones 
which optimize the ratio between the primal and dual values. If an algorithm with output 
$x^{*}$ satisfies $P(x^{*}) \leq r N(x^{*},\lambda,\gamma)$ for some factor $r$ then 
the algorithm is $r$-competitive.

\section{Minimizing Total Energy plus Lost Values}	\label{sec:energy+values}

\paragraph{The problem.}
We are given a machine with a convex energy power $P$ and
jobs arrive over time. Each job $j$ is released at time $r_{j}$, has deadline $d_{j}$, processing volume 
$p_{j}$ and a value $a_{j}$. Jobs could be executed preemptively and at any time $t$, the scheduler 
has to choose a set of \emph{pending} jobs (i.e., $r_{j} \leq t < d_{j}$) and a machine speed $s(t)$ 
in order to process such jobs. The \emph{energy cost} of a schedule is 
$\int_{0}^{\infty} P(s(t))dt$. Typically,
$P(z) = z^{\alpha}$ for some constant $\alpha \geq 1$. The objective of the problem
is to minimize energy cost plus the \emph{lost value} --- which is the total value of uncompleted jobs.

\paragraph{Formulation.} Let $x_{j}$ and $y_{j}$ be variables indicating whether job $j$ is completed or it is not.
We denote variable $s_{j}(t)$ as the speed that the machine processes job $j$ at time $t$.  
The problem could be relaxed as the following convex program.
 
  \begin{alignat*}{3}
    \text{min} \quad \int_{0}^{\infty} P(s(t)) & dt + \sum_{j} a_{j} y_{j}  \\ 
	\text{subject to} \qquad \qquad s(t) &= \sum_{j} s_{j}(t)	\qquad &\forall t \\
						 x_{j} + y_{j} &\geq 1 \qquad &\forall j \\
					         \int_{r_{j}}^{d_{j}} s_{j}(t)dt & \geq  p_{j}x_{j}   \qquad &\forall j\\
					         x_{j}, y_{j}, s_{j}(t) &\geq 0 \qquad &\forall j, t
  \end{alignat*}
In the relaxation, the second constraint indicates that either job $j$ is completed or it is not. 
The third constraint guarantees the necessary amount of work done in order to complete job $j$.  

Applying the framework, we have the following dual.
$$
\max \int_{0}^{\infty} P\biggl(\sum_{j} v^{*}_{j}(t) \biggl)dt - \sum_{j} \lambda_{j} \int_{r_{j}}^{d_{j}} v^{*}_{j}(t)dt 
	+ \sum_{j} \gamma_{j}
$$
subject to
\begin{enumerate}
	\item For any job $j$, $\gamma_{j} \leq p_{j}\lambda_{j}$. 
	Moreover,  if $x^{*}_{j} > 0$ then $\gamma_{j} = p_{j}\lambda_{j}$. 
	\item For any job $j$, $\gamma_{j} \leq a_{j}$ and  if $y^{*}_{j} > 0$ then $\gamma_{j} = a_{j}$.
	\item For any job $j$ and any $t \in [r_{j},d_{j}]$, it holds that 
	$\lambda_{j} \leq P'(\sum_{k} v^{*}_{k}(t))$. 
	Particularly, if $v^{*}_{j}(t) > 0$ then $\lambda_{j} = P'(\sum_{k} v^{*}_{k}(t))$. 
\end{enumerate}
Note that $v^{*}_{j}(t)$ is not equal to $s^{*}_{j}(t)$ (the machine speed on job $j$ according to our algorithm) 
but it is a function depending on $s^{*}_{j}(t)$. That is the reason we use $v^{*}_{j}(t)$ instead of $s^{*}_{j}(t)$.
We will choose $v^{*}_{j}(t)$'s in order to optimize the competitive ratio. To simplify the notation, we drop out 
the star symbol in the superscript of every variable (if one has that). 

\paragraph{Algorithm.} The dual constraints naturally leads to the following algorithm. We first describe
informally the algorithm. 
In the algorithm, we maintain a variable $u_{j}(t)$ representing the 
\emph{virtual} machine speed on job $j$. The virtual speed on job $j$ means 
that job $j$ will be processed with that speed \emph{if} it is accepted; otherwise, 
the real speed on $j$ will be set to 0. 
Consider the arrival of job $j$. Observe that by the third dual constraint, 
we should always increase the machine speed on job $j$
at $\arg \min P'(v(t))$ in order to increase $\lambda_{j}$.
Hence, at the arrival of a job $j$, increase continuously the \emph{virtual} 
speed $u_{j}(t)$ of job $j$ at $\arg \min P'(v(t))$ for $r_{j} \leq t \leq d_{j}$.
Moreover, function $v(t)$ is also simultaneously 
updated as a function of $u(t) = \sum_{k \preceq j} u_{k}(t)$ 
according to a system of differential equations
(\ref{eq:PD-throughput+energy}) in order to optimize the competitive ratio.
The iteration on job $j$ terminates whether one of the first two constraints becomes tight. 
If the first one holds, then accept the job and set the real speed equal to 
the virtual one. Otherwise, reject the job.

Define $Q(z) := P(z) - zP'(z)$.
Consider the following system of differential equations with 
boundary conditions: $Q(v) = 0$ if $u = 0$.
\begin{align}	\label{eq:PD-throughput+energy}
	Q'(v) \frac{dv}{du} + P'(v) &\geq \frac{P'(u)}{r},  \notag \\ 
	(r-1) P'(u) + r Q'(v) \frac{dv}{du} &\geq 0, \\
	 \frac{dv}{du} &> 0, \notag
\end{align}
where $r$ is some constant. Let $r^{*} \geq 1$ be a smallest constant such that the system has 
a solution. 

The formal algorithm is given in Algorithm~\ref{algo:energy+values}. 

\begin{algorithm}[htbp]
\begin{algorithmic}[1] 
\STATE Initially, set $s(t), s_{j}(t)$, $u_{j}(t)$, $v(t)$ and $v_{j}(t)$ equal to 0 for every $j$.
\STATE Let $r^{*} \geq 1$ be the smallest constant such that  
	(\ref{eq:PD-throughput+energy}) has a solution. During the algorithm, keep $v(t)$
	as a solution of (\ref{eq:PD-throughput+energy}) with constant $r^{*}$ and $u(t) = \sum_{j}u_{j}(t)$
	for every time $t$.
\FOR{a job $j$ arrives}
	\STATE Initially, $u_{j}(t) \gets 0$.
	\WHILE{$\int_{r_{j}}^{d_{j}} u_{j}(t)dt < p_{j}$ \textbf{and} $\lambda_{j} p_{j} < a_{j}$}
		\STATE Continuously increase $u_{j}(t)$ at $\arg \min P'(v(t))$  for $r_{j} \leq t \leq d_{j}$
		and update $u(t) \gets \sum_{k \neq j} u_{k}(t) + u_{j}(t)$ and $v(t)$ (as a function of $u(t)$) and
		$\lambda_{j} \gets \min_{r_{j} \leq t \leq d_{j}} P'(v(t))$ simultaneously. 
	\ENDWHILE
	\STATE Set $v_{j}(t) \gets v(t) - \sum_{k \prec j} v_{k}(t)$.
	\IF{$\lambda_{j} p_{j} = a_{j}$ \textbf{and} $\int_{r_{j}}^{d_{j}} u_{j}(t) dt < p_{j}$}
		\STATE Reject job $j$.
		\STATE Set $\gamma_{j} \gets a_{j}$.
	\ELSE
		\STATE Accept job $j$.
		\STATE Set $s_{j}(t) \gets u_{j}(t)$, $s(t) \gets s(t) + s_{j}(t)$ 
			     and $\gamma_{j} \gets \lambda_{j}p_{j}$.
	\ENDIF
\ENDFOR
\end{algorithmic}
\caption{Minimizing the consumed energy plus lost values.}
\label{algo:energy+values}
\end{algorithm}
%

In the algorithm, machine $i$ processes accepted job $j$ with speed $s_{j}(t)$ at time $t$.
As the algorithm completes all accepted jobs, it is equivalent to state that the machine
processes accepted jobs in the earliest deadline first fashion with speed $s(t)$ at time $t$.

By the algorithm, the dual variables are feasible. In the following
we bound the values of the primal and dual objectives.

\begin{lemma}		\label{lem:throughput+energy}
It holds that 
$$
\int_{0}^{\infty} P(s(t))dt + \sum_{j} a_{j} y_{j}  
\leq r^{*} \left( \int_{0}^{\infty} P\biggl(\sum_{j} v_{j}(t) \biggl)dt 
	- \sum_{j} \lambda_{j} \int_{r_{j}}^{d_{j}} v_{j}(t)dt  + \sum_{j} \gamma_{j} \right)
$$
\end{lemma}
\begin{proof}
By the algorithm, $\lambda_{j} = P'(v(t))$ at every time $t$ such that $v_{j}(t) > 0$
for every job $j$. Hence, it is sufficient to show that 
\begin{equation}	\label{ineq:throughput+energy-1}
\int_{0}^{\infty} P(s(t))dt + \sum_{j} a_{j} y_{j}  
\leq r^{*} \left( \int_{0}^{\infty} Q\biggl(\sum_{j} v_{j}(t) \biggl)dt  + \sum_{j} \gamma_{j} \right)
\end{equation}
where recall $Q(z) = P(z) - zP'(z)$.

We will prove the inequality (\ref{ineq:throughput+energy-1}) 
by induction on the number of jobs in the instance. For the base case where there is 
no job, the inequality holds trivially. Suppose that the inequality holds before the arrival of a job $j$.
In the following, we consider different cases.

\paragraph{Job $j$ is accepted.} 
Consider any moment $\tau$ in the while loop related to job $j$. We emphasize that $\tau$ is a moment 
in the execution of the algorithm, not the one in the time axis $t$.
Suppose that at moment $\tau$, an amount $du_{j}(t)$ is increased (allocated) at $t$. Note that 
$du_{j}(t) = du(t)$ as $u(t) = \sum_{j} u_{j}(t)$.
As $j$ is accepted, $y_{j} = 0$ and 
the increase at $\tau$ in the left hand-side of (\ref{ineq:throughput+energy-1})
is $P'(u(t))du(t)$ 

Let $v(t_{1},\tau_{1})$ be the value of $v(t_{1})$ at moment $\tau_{1}$ in 
the while loop.  By the algorithm, the dual variable $\gamma_{j}$ satisfies
\begin{align*}
\gamma_{j} = \lambda_{j} p_{j} 
&\geq \int_{\tau_{1}} \biggl( \int_{r_{j}}^{d_{j}} u_{j}(t_{2})dt_{2} \biggl) 
			\min_{r_{j} \leq t_{1} \leq d_{j}} P'(v(t_{1},\tau_{1}))d\tau_{1} \\
&= \int_{\tau_{1}} \int_{r_{j}}^{d_{j}} u_{j}(t_{2})
			\min_{r_{j} \leq t_{1} \leq d_{j}} P'(v(t_{1},\tau_{1})) dt_{2} d\tau_{1}
\end{align*}
where the inequality is due to the fact that at the end of the while loop, 
$\int_{r_{j}}^{d_{j}} u_{j}(t) dt = p_{j}$ (by the loop condition) and $P'$ is increasing. 
Therefore, at moment $\tau$, 
$d\gamma_{j} \geq \min_{r_{j} \leq t_{1} \leq d_{j}}P'(v(t_{1},\tau)) du_{j}(t) = P'(v(t))du(t)$
where the equality follows since $t \in \arg \min_{r_{j} \leq t_{1} \leq d_{j}} P'(v(t_{1},\tau))$.
Hence, the increase in the right hand-side of (\ref{ineq:throughput+energy-1}) is at least 
$r^{*}[Q'(v(t))dv(t) + P'(v(t))du(t)]$.

Due to the system of inequations~(\ref{eq:PD-throughput+energy}) and the choice of $r^{*}$, 
at any moment in the execution of the 
algorithm, the increase in the left hand-side of (\ref{ineq:throughput+energy-1}) is 
at most that in the right hand-side. Thus, the induction step follows.

\paragraph{Job $j$ is rejected.}
If $j$ is rejected then $y_{j} = 1$ and so the increase in the left hand-side of (\ref{ineq:throughput+energy-1}) 
is $a_{j}$. Moreover, by the algorithm $\gamma_{j} = a_{j}$.
So we need to prove that after the iteration of the for loop 
on job $j$, it holds that 
$(r^{*}-1) a_{j} + r^{*} \int_{0}^{\infty} \Delta Q(v(t)) dt \geq 0$.
As $j$ is rejected, 
$$
a_{j} = \lambda_{j}p_{j} > \int_{r_{j}}^{d_{j}} \min_{r_{j} \leq t_{1} \leq d_{j}} P'(v(t_{1})) u_{j}(t) dt
$$
Therefore, it is sufficient to prove that 
\begin{equation}	\label{ineq:energy+values-2}
(r^{*} - 1)\int_{r_{j}}^{d_{j}} \min_{r_{j} \leq t_{1} \leq d_{j}} P'(v(t_{1})) u_{j}(t) dt
+ r^{*} \int_{0}^{\infty} \Delta Q(v(t)) dt \geq 0
\end{equation}
Before the iteration of the while loop, the left-hand side of (\ref{ineq:energy+values-2}) 
is $0$. Similar as the analysis of the previous case, during the execution of the algorithm
the increase rate of the left-hand side is 
$(r^{*} - 1) P'(v(t)) du(t) + r^{*} Q'(v(t)) dv(t)$, which is non-negative by
equation~(\ref{eq:PD-throughput+energy}). Thus, inequality (\ref{ineq:energy+values-2})
holds.

By both cases, the lemma follows. 
\end{proof}

\begin{theorem}
The algorithm is $r^{*}$-competitive. Particularly,
if the energy power function $P(z) = z^{\alpha}$ then the algorithm 
is $\alpha^{\alpha}$-competitive
\end{theorem}
\begin{proof}
The theorem follows by the framework
and Lemma \ref{lem:throughput+energy}.

If the power energy function $P(z) = z^{\alpha}$ 
then $r^{*} = \alpha^{\alpha}$ and $u(t) = v(t)/\alpha$ satisfy the system 
(\ref{eq:PD-throughput+energy}). Thus, the algorithm is $\alpha^{\alpha}$-competitive.
\end{proof}

\section{Maximizing the Total Value minus Energy}		\label{sec:values-energy}

\paragraph{The problem.}
We are given unrelated machines and jobs arrive over time. 
Each job $j$ is released at time $r_{j}$, has deadline $d_{j}$, a value $a_{j}$ 
and processing volume $p_{ij}$ if it is executed on machine $i$. 
Jobs could be executed preemptively but migration is not allowed, i.e., no job 
could be executed in more than one machine. 
At a time $t$, the scheduler has to choose a set of \emph{pending} jobs (i.e., $r_{j} \leq t < d_{j}$)  
to be processed on each machine, and 
the speeds $s_{i}(t)$'s for every machine $i$ to execute such jobs.  
The energy cost is $\sum_{i} \int_{0}^{\infty} P(s_{i}(t))dt$ where $P$ is a given convex power function.
The objective now is to maximize the total value of completed jobs minus the energy cost. 

We first give some idea about the difficulty of the problem even on a single machine. 
Assume that the adversary releases a job with small value but with a high energy 
cost in order to complete the job. One has to execute the job since otherwise the profit would be zero.
However, at the moment an algorithm nearly completes the job, the adversary releases
other job with much higher value and a reasonable energy demand. One need to switch immediately 
to the second job since otherwise either a high value is lost or the energy consumption becomes 
too much. It means that all energy spending on the first job is lost without any gain. 
Based on this idea, \cite{PruhsStein10:How-to-Schedule-When} has shown that 
without resource augmentation, the competitive ratio 
is unbounded. 

In this section, we consider the problem with resource augmentation, meaning that 
with the same speed $z$ the energy power for the algorithm is $P((1-\epsilon)z)$, 
whereas the one for the adversary is $P(z)$. 
Let $\epsilon(P) > 0$ be the smallest constant such that $zP'((1-\epsilon(P))z) \leq P(z)$ for 
all $z > 0$. For the typical energy power $P(z) = z^{\alpha}$, $\epsilon(P) = 1 - \alpha^{-1/\alpha}$
which is closed to 0 for $\alpha$ large.

\paragraph{Formulation.} Let $x_{ij}$ be variable indicating whether job $j$ is completed
in machine $i$. 
Let $s_{ij}(t)$ be the variable representing the speed
that machine $i$ processes job $j$ at time $t$.   
The problem could be formulated as the following convex program.
 
  \begin{alignat*}{3}
    \text{max} \quad  \sum_{i,j} a_{j} x_{ij} - & \sum_{i}\int_{0}^{\infty} P((1-\epsilon)s_{i}(t)) dt \\
	\text{subject to} \qquad \qquad s_{i}(t) &= \sum_{j} s_{ij}(t)	\qquad &\forall i, t \\
					\sum_{i} x_{ij} &\leq 1 \qquad &\forall j \\
					\int_{r_{j}}^{d_{j}} s_{ij}(t)dt & \geq p_{ij}x_{ij}   
											\qquad &\forall i,j \\
					     x_{ij}, s_{ij}(t) &\geq 0 \qquad &\forall i,j,t
  \end{alignat*}
Note that in the objective, by resource augmentation the consumed energy 
is $\sum_{i}\int_{0}^{\infty} P((1-\epsilon)s_{i}(t)) dt$. 
Applying the framework, we have the following dual.
$$
\min \sum_{j} \gamma_{j} - \sum_{i} \int_{0}^{\infty} P\biggl(\sum_{j} u^{*}_{ij}(t) \biggl)dt 
+ \sum_{i,j} \lambda_{ij} \int_{r_{j}}^{d_{j}} u^{*}_{ij}(t)dt 
$$
subject to
\begin{enumerate}
	\item For any machine $i$ and any job $j$, 
	$\gamma_{j} + p_{ij}\lambda_{ij} \geq a_{j}$. 
	\item For any machine $i$, %
	any job $j$ and any $t \in [r_{j},d_{j}]$, 
	$\lambda_{ij} \leq P'((1-\epsilon)\sum_{k} u^{*}_{ik}(t))$ where the sum is taken over all jobs $k$ released before $j$,
	i.e., $k \preceq j$. 
	Particularly, if $u^{*}_{ij}(t) > 0$ then $\lambda_{ij} = P'((1-\epsilon)\sum_{k \preceq j} u^{*}_{ik}(t))$. 
\end{enumerate}

Similar as in the previous section, the constraints naturally lead to Algorithm~\ref{algo:values-energy}.
In the algorithm and the analysis, to simplify the notation we drop out 
the star symbol in the superscript of every variable (if one has that).

\begin{algorithm}[ht]
\begin{algorithmic}[1] 
\STATE Initially, set $s(t)$ and $u(t)$ equal to 0.
\STATE The algorithm always runs accepted jobs with speed $s(t)$ in the earliest deadline 
	first fashion. 
\FOR{a job $j$ arrives}
	\STATE Initially, $u_{ij}(t) \gets 0$ for every $t$ and let $\mathcal{I}$ be the set of all machines,
		$\mathcal{I'} \gets \emptyset$.
	\WHILE{$\mathcal{I} \neq \emptyset$}
		\STATE For every $i \in \mathcal{I}$, 
		increase $u_{ij}(t)$ at $\arg \min P'(u_{i}(t))$ in the continuous manner for $r_{j} \leq t \leq d_{j}$
		and update $u_{i}(t) = \sum_{k \neq j} u_{ik}(t) + u_{ij}(t)$ and 
		$\lambda_{ij} = \min_{r_{j} \leq t \leq d_{j}} P'((1-\epsilon)u_{i}(t))$ simultaneously. 
		\IF{$\lambda_{ij} p_{ij} = a_{j}$ \textbf{and} $\int_{r_{j}}^{d_{j}} u_{ij}(t)dt < p_{j}$ for some machine $i$}
			\STATE $\mathcal{I} \gets \mathcal{I} \setminus \{i\}$.
		\ENDIF
		\IF{$\lambda_{ij} p_{ij} < a_{j}$ \textbf{and} $\int_{r_{j}}^{d_{j}} u_{ij}(t)dt = p_{j}$ for some machine $i$}
			\STATE $\mathcal{I'} \gets \mathcal{I'} \cup \{i\}$ and  $\mathcal{I} \gets \mathcal{I} \setminus \{i\}$.
		\ENDIF
	\ENDWHILE
	\IF{$\mathcal{I'} = \emptyset$}
		\STATE Reject job $j$ and set $\gamma_{j} \gets 0$ (note that $p_{ij}\lambda_{ij} = a_{j} ~\forall i$).
	\ELSE
		\STATE Let $i = \arg \min_{i' \in \mathcal{I'}} p_{i'j}\lambda_{i'j}$.
		\STATE Accept and assign job $j$ to machine $i$, i.e., $x_{ij} = 1$.
		\STATE Set $s_{ij}(t) \gets u_{ij}(t)$, $s_{i}(t) \gets s_{i}(t) + s_{ij}(t)$
				and $\gamma_{j} \gets a_{j} - \lambda_{j}p_{j}$.
	\ENDIF
\ENDFOR
\end{algorithmic}
\caption{Minimizing the throughput minus consumed energy.}
\label{algo:values-energy}
\end{algorithm}

\begin{lemma}
Dual variables constructed by Algorithm~\ref{algo:values-energy} are feasible.
\end{lemma}
\begin{proof}
By the algorithm (line 6), $\lambda_{ij} \leq P'((1-\epsilon)\sum_{k} u_{ik}(t))$ 
where the sum is taken over all jobs $k$ released before $j$ ($k \preceq j$) and
if $u_{ij}(t) > 0$ then $\lambda_{ij} = P'((1-\epsilon)\sum_{k \preceq j} u_{ik}(t))$. 
Consider the first constraint. If $j$ is rejected then $p_{ij} \lambda_{ij} = a_{j}$ for every machine $i$.
Otherwise, by the assignment (line 17),
it always holds that $\gamma_{j} + p_{ij}\lambda_{ij} \geq a_{j}$ for every $i,j$.
\end{proof}

In the following we bound the values of the primal and dual objectives 
in the resource augmentation model.

\begin{lemma}		\label{lem:throughput-energy}
For every $\epsilon \geq \epsilon(P)$, it holds that 
$$ 
\frac{1}{\epsilon} \biggl( \sum_{i,j} a_{j} x_{ij}  - \sum_{i}\int_{0}^{\infty} P((1-\epsilon)s_{i}(t))dt \biggl)
\geq \sum_{j} \gamma_{j} - \sum_{i} \int_{0}^{\infty} P\biggl(\sum_{j} u_{ij}(t) \biggl)dt 
+ \sum_{i,j} \lambda_{ij} \int_{r_{j}}^{d_{j}} u_{ij}(t)dt 
$$
\end{lemma}
\begin{proof}
We have 
\begin{align*}
\sum_{i,j} &\lambda_{ij} \int_{r_{j}}^{d_{j}} u_{ij}(t)dt - \sum_{i} \int_{0}^{\infty} P(u_{i}(t))dt \\
&\leq \sum_{i} \int_{0}^{\infty} \biggl( \max_{j: r_{j} \leq t \leq d_{j}} \lambda_{ij} \biggl) \sum_{j} u_{ij}(t)dt 
	- \sum_{i} \int_{0}^{\infty} P(u_{i}(t))dt \\
&\leq \sum_{i} \int_{0}^{\infty} P'((1-\epsilon)u_{i}(t)) u_{i}(t)dt 
	- \sum_{i} \int_{0}^{\infty} P(u_{i}(t))dt
\leq 0
\end{align*}
In the second inequality, recall that $\sum_{j} u_{ij}(t) = u_{i}(t)$
and by the algorithm, $\lambda_{ij} = \min_{r_{j} \leq \tau \leq d_{j}} P'((1-\epsilon)u_{i}(\tau)) \leq P'((1-\epsilon)u_{i}(t))$
for any $r_{j} \leq t \leq d_{j}$. 
The last inequality follows by the 
definition of $\epsilon(P)$. Therefore, it is sufficient to prove that 
\begin{equation}	\label{eq:throughput-energy-1}
\frac{1}{\epsilon} \biggl( \sum_{i,j} a_{j} x_{ij}  - \sum_{i}\int_{0}^{\infty} P((1-\epsilon)s_{i}(t))dt \biggl)
\geq \sum_{j} \gamma_{j}
\end{equation}
We prove inequality (\ref{eq:throughput-energy-1})
by induction on the number of released jobs in the instance. For the base case where there is 
no job, the inequality holds trivially. Suppose that the inequality holds before the arrival of a job $j$.

If $j$ is rejected then $x_{ij} = 0, s_{ij}(t) = 0$ for every $i,t$ and $\gamma_{j} = 0$. 
Therefore, the increases in both side of inequality (\ref{eq:throughput-energy-1}) are 0.
Hence, the induction step follows. 

In the rest, assume that $j$ is accepted and let $i$ be the machine 
to which $j$ is assigned. 
We have
\begin{align*}
& \int_{0}^{\infty} \biggl[ P\biggl((1-\epsilon)u_{i}(t) \biggl) - P\biggl((1-\epsilon)(u_{i}(t) - u_{ij}(t)) \biggl) \biggl] dt
\leq (1-\epsilon) \int_{0}^{\infty} P'\biggl( (1-\epsilon)u_{i}(t) \biggl) u_{ij}(t) dt \\
&\leq (1-\epsilon) \int_{0}^{\infty} P'(u_{i}(t)) u_{ij}(t) dt 
= (1-\epsilon) \lambda_{ij} \int_{0}^{\infty} u_{ij}(t) dt 
= (1-\epsilon) \lambda_{ij} p_{ij} \leq (1-\epsilon) a_{j}
\end{align*}
The first inequality is due to the convexity of $P$
$$
P\biggl( (1-\epsilon)(u_{i}(t) - u_{ij}(t)) \biggl) 
	\geq P\biggl( (1-\epsilon)u_{i}(t) \biggl) - (1-\epsilon)u_{ij}(t)P'\biggl( (1-\epsilon)u_{i}(t) \biggl).
$$ 
The second inequality holds because $P'$ is increasing. 
The first equality follows since
$u_{ij}(t) \neq 0$ only at $t$ such that $P'(u_{i}(t)) = \lambda_{ij}$ (by the algorithm).     
The last inequality is due to the loop condition in the algorithm.
Thus, at the end of the iteration (related to job $j$) in the for loop,
the increase in the left hand side of inequality (\ref{eq:throughput-energy-1}) is
$$
\frac{1}{\epsilon} 
\biggl( a_{j}x_{ij} - \int_{0}^{\infty} \biggl[ P\biggl((1-\epsilon)u_{i}(t) \biggl) 
					- P\biggl((1-\epsilon)(u_{i}(t) - u_{ij}(t)) \biggl) \biggl] dt \biggl)
\geq \frac{1}{\epsilon} \biggl( a_{j} - (1-\epsilon)a_{j} \biggl) = a_{j}
$$ 
Besides, the increase in the right hand-side of inequality (\ref{eq:throughput-energy-1}) is 
$\gamma_{j} \leq a_{j}$.
Hence, the induction step follows; so does the lemma.
\end{proof}

\begin{theorem}
The algorithm is $(1+\epsilon)$-augmentation, $1/\epsilon$-competitive
for $\epsilon \geq \epsilon(P)$. 
\end{theorem}
\begin{proof}
By resource augmentation, with the same speed $z$ the energy power 
for the algorithm is $P((1-\epsilon)z)$, 
whereas the one for the adversary is $P(z)$. So by 
Lemma~\ref{lem:throughput-energy}, the theorem follows.
\end{proof}

Note that the result could be generalized for \emph{heterogeneous} machines 
where the energy power functions are different. In this case, one needs to consider 
$\epsilon \geq \max_{i} \epsilon(P_{i})$.

\section{Energy Minimization in Speed Scaling with Power Down Model}		\label{sec:4S-energy}
\paragraph{The problem.} 
We are given a single machine that could be transitioned into 
a sleep state or an active state.  Each transition from the sleep state to 
the active state costs $A > 0$, which is called the \emph{wake-up cost}. 
Jobs arrive online, each job has a released time $r_{j}$, a deadline $d_{j}$, 
a processing volume $p_{j}$ and could be processed preemptively. In the problem,
all jobs have to be completed. In the sleep state, the energy consumption of the machine is 0.
In the active state, the power energy consumption at time $t$ is $P(s(t)) = s(t)^{\alpha} + g$
where $\alpha \geq 1$ and $g \geq 0$ are constant. Thus, the consumed energy of the machine 
in active state is $\int_{0}^{\infty} P(s(t))dt$, that can be decomposed into 
\emph{dynamic energy} $\int_{0}^{\infty} s(t)^{\alpha}dt$ and 
\emph{static energy} $\int_{0}^{\infty} g dt$ (where the integral is taken over $t$ at which 
the machine is in active state). At any time $t$, 
the scheduler has to decide the state of the machine and the speed if the machine is in active state 
in order to execute and complete all jobs. 
The objective is to minimize the total energy --- the consumed energy in active state plus the 
wake-up energy. 

\paragraph{Formulation.} In a mathematical program for the problem, we need to incorporate an information 
about the machine states and the transition cost from the sleep state to the active one.  
Here we make use of 
the properties of the Heaviside step function and the Dirac delta function to encode 
the machine states and the transition cost. Recall that the Heaviside step 
function $H(t) = 0$ if $t < 0$ and $H(t) = 1$ if $t \geq 0$. Then $H(t)$
is the integral of the Dirac delta function $\delta$ (i.e., $H' = \delta$) and it holds that 
$\int_{-\infty}^{+\infty} \delta(t)dt = 1$. Now let $F(t)$ be a function
indicating whether the machine is in active state at time $t$, i.e., $F(t) = 1$ if 
the machine is active at $t$ and $F(t) = 0$ if it is in the sleep state. 
Assume that the machine initially is in the sleep state. Then
$A \int_{0}^{+\infty} |F'(t)|dt$ equals twice the transition cost
of the machine (a transition from the active state to the sleep state costs 0 while 
in $A \int_{0}^{+\infty} |F'(t)|dt$, it costs $A$).

Let $s_{j}(t)$ be variable representing the machine speed on job $j$ at time $t$.   
The problem could be formulated as the following (non-convex) program.
 
  \begin{alignat*}{3}
    \text{min} \quad  \int_{0}^{\infty} P\biggl(\sum_{j} s_{j}(t) \biggl)F(t) &dt + \frac{A}{2} \int_{0}^{+\infty} |F'(t)|dt \\
	\text{subject to} \qquad \qquad     \int_{r_{j}}^{d_{j}} s_{j}(t) F(t)dt &\geq p_{j} \qquad &\forall j \\
					     s_{j}(t) \geq 0, &~F(t) \in \{0,1\} \qquad &\forall j,t
  \end{alignat*}
The first constraint ensures that every job $j$ will be fully processed during $[r_{j},d_{j}]$.
Moreover, each time a job is executed, the machine has to be in the active state. 
Note that we do not relax the variable $F(t)$. The objective function consists of 
the energy cost during the active periods and the wake-up cost. 

\subsection{Speed Scaling without Wake-Up Cost}
The problem without wake-up cost ($A = 0$) has been extensively studied. 
We reconsider the problem throughout our primal-dual approach. In case $A = 0$,
the machine is put in active state whenever there is some pending job (thus the function
$F(t)$ is useless and could be removed from the formulation).  
In this case, the relaxation above becomes a convex program. Applying the framework and by the same
observation as in previous sections, we derive the following algorithm.

At the arrival of job $j$, increase continuously $s_{j}(t)$ at $\arg \min P'(s(t))$ for $r_{j} \leq t \leq d_{j}$
and update simultaneously $s(t) \gets s(t) + s_{j}(t)$ until $\int_{r_{j}}^{d_{j}} s_{j}(t')dt' = p_{j}$.

It turns out that the machine speed $s(t)$ of the algorithm equals $\max_{t' > t} V(t,t')/(t' - t)$
where $V(t,t')$ is the remaining processing volume of jobs arriving at or before $t$
with deadline in $(t,t']$. So the algorithm is indeed algorithm \textsc{Optimal Available} \cite{YaoDemers95:A-Scheduling-Model}
that is $\alpha^{\alpha}$-competitive \cite{BansalKimbrel07:Speed-scaling}.
However, the primal-dual view of the algorithm gives more insight and that is useful 
in the general energy model (see Lemma~\ref{lem:general-energy}). 

\subsection{Speed Scaling with Wake-Up Cost}

\paragraph{The Algorithm.}
Define the \emph{critical} speed $s^{c} = \arg \min_{s > 0} P(s)/s$.
In the algorithm, the machine speed is always at least $s^{c}$ if it executes some job.
 
Initially, set $s(t)$ and $s_{j}(t)$ equal 0 for every time $t$ and jobs $j$.
If a job is released then it is marked as \emph{active}. Intuitively, a job is active 
if its speed $s_{j}(t)$ has not been settled yet. 
Let $\tau$ be the current moment. Consider currently active jobs in the earliest deadline 
first (EDF) order. Increase continuously $s_{j}(t)$ at $\arg \min P'(s(t))$ for $r_{j} \leq t \leq d_{j}$
and update simultaneously $s(t) \gets s(t) + s_{j}(\tau)$ until $\int_{r_{j}}^{d_{j}} s_{j}(t')dt' = p_{j}$.
Now consider different states of the machine at the current time $\tau$. We distinguish 
three different states: (1) in \emph{working state} the machine is active and 
is executing some jobs; (2) in \emph{idle state} the machine is active but 
its speed equals 0; and (3) in \emph{sleep state} the machine is inactive.   
\begin{description}
	\item[In working state.] If $s(\tau) > 0$ then keep 
	process jobs with the earliest deadline by speed $\max\{s(\tau),s^{c}\}$. 
	Mark all currently pending jobs as inactive.
	If $s(\tau) = 0$, switch to the idle state.
	\item[In idle state.] If $s(\tau) \geq s^{c}$ then switch to the working state. \\
	If $s^{c} > s(t) > 0$. Mark all currently pending jobs as active. 
	Intuitively, we delay such jobs until some moment where the machine has to run
	at speed $s^{c}$ in order to complete these jobs (assuming that there is no new job released).\\
	Otherwise, if the total duration of idle state from the last wake-up equals $A/g$ then switch 
	to the sleep state.
	\item[In sleep state.] If $s(t) \geq s^{c}$ then switch to the working state.
\end{description}

In the rest, we denote $s^{*}(t)$ as the machine speed at time $t$
by the algorithm. Moreover, let $s^{*}_{j}(t)$ be the speed of the algorithm
on job $j$ at time $t$.

\paragraph{Analysis.} The Lagrangian dual is
$
\max_{\lambda \geq 0} \min_{s,F} L(s,F,\lambda)
$ 
where the minimum is taken over  $(s,F)$ feasible solutions of the primal and 
$L$ is the following Lagrangian function 
\begin{align*}
L(s,F,\lambda) &=
	\int_{0}^{\infty} P\biggl(\sum_{j} s_{j}(t) \biggl)F(t)dt + \frac{A}{2} \int_{0}^{+\infty} |F'(t)|dt
	+ \sum_{j} \lambda_{j} \biggl(p_{j} -  \int_{r_{j}}^{d_{j}} s_{j}(t) F(t)dt \biggl) \\
& \geq \sum_{j} \lambda_{j}p_{j} 
	- \sum_{j}  \int_{r_{j}}^{d_{j}} s_{j}(t) F(t) \biggl( \lambda_{j} - \frac{{P}(s(t))}{s(t)}\biggl) 
		\one_{\{s(t) > 0\}}  \one_{\{F(t) = 1\}}dt 
	+ \frac{A}{2} \int_{0}^{+\infty} |F'(t)|dt
\end{align*}
where $s(t) = \sum_{j} s_{j}(t)$. 

By weak duality, the optimal value of the primal is always larger than the one of the 
corresponding Lagrangian dual. In the following, we bound the Lagrangian dual value
in function of the algorithm cost and derive the competitive 
ratio via the dual-fitting approach.

\paragraph{Dual variables} Let $0 < \beta \leq 1$ be 
some constant to be chosen later. For jobs $j$ such that 
$s^{*}(t) > 0$ for every $t \in [r_{j},d_{j}]$,
define $\lambda_{j}$ such that $\lambda_{j}p_{j}/\beta$ equals the marginal 
increase of the \emph{dynamic} energy due to the arrival of job $j$. For jobs $j$ such that 
$s^{*}(t) = 0$ for some moment $t \in [r_{j},d_{j}]$,
define $\lambda_{j}$ such that $\lambda_{j}p_{j}$ equals the marginal 
increase of the \emph{dynamic and static} energy due to the arrival of job $j$.

\begin{lemma}		\label{lem:general-energy}
Let $j$ be an arbitrary job. 
\begin{enumerate}
	\item If $s^{*}(t) > 0$ for every $t \in [r_{j},d_{j}]$ then
	$\lambda_{j} \leq \beta P'(s^{*}(t))$ 
	for every $t \in [r_{j},d_{j}]$.
	\item Moreover, if 
	$s^{*}(t) = 0$ for some $t \in [r_{j},d_{j}]$ then $\lambda_{j} = P(s^{c})/s^{c}$.
\end{enumerate}
\end{lemma}
\begin{proof}
We prove the first claim.
For any time $t$, speed $s^{*}(t)$ is non-decreasing as long as new jobs arrive. Hence,
it is sufficient to prove the claim assuming that no other job is released after $j$. So
$s^{*}(t)$ is the machine speed after the arrival of $j$. The marginal increase in the dynamic energy 
due to the arrival of $j$ could be written as 
\begin{align*}
 \frac{1}{\beta}\lambda_{j}p_{j} &= \int_{r_{j}}^{d_{j}} \biggl( P(s^{*}(t)) - P(s^{*}(t) - s^{*}_{j}(t) \biggl)dt
  	\leq   \int_{r_{j}}^{d_{j}} P'(s^{*}(t)) s^{*}_{j}(t))dt  \\
	&= \min P'(s^{*}(t)) \int_{r_{j}}^{d_{j}} s^{*}_{j}(t)dt 
	  = \min P'(s^{*}(t)) p_{j}
\end{align*}
where $\min P'(s^{*}(t))$ is taken over $t \in [r_{j},d_{j}]$ such that $s^{*}_{j}(t) > 0$.
The inequality is due to the convexity of $P$ and the second equality follows by the algorithm. 
Moreover, $\min_{r_{j} \leq t \leq d_{j}} P'(s^{*}(t)) \leq P'(s^{*}(\tau))$ for any $\tau \in [r_{j},d_{j}]$; so the lemma 
follows. 

We are now showing the second claim. 
By the algorithm, the fact that $s^{*}(t) = 0$ for some $t \in [r_{j},d_{j}]$ means that
job $j$ will be processed at speed $s^{c}$ in some interval $[a,b] \subset [r_{j},d_{j}]$
(assuming that no new job is released after $r_{j}$).  
The marginal increase in the energy is $P(s^{c})(b-a)$ while $p_{j}$ could be expressed 
as $s^{c}(b-a)$. Therefore, $\lambda_{j} = P(s^{c})/s^{c}$.
\end{proof}

\begin{theorem}
The algorithm has competitive ratio at most $\max\{4,\alpha^{\alpha}\}$.
\end{theorem}
\begin{proof}
Let $E^{*}_{1}$ be the dynamic energy of the algorithm 
schedule, i.e., $E^{*}_{1} = \int_{0}^{\infty} [P(s^{*}(t)) - P(0)] dt \leq \sum_{j} \lambda_{j}p_{j}/\beta$
due to the definition of $\lambda_{j}$'s and $0 < \beta \leq 1$. 
Moreover, let $E^{*}_{2}$ be the static energy plus the wake-up energy of the algorithm, i.e.,
$E^{*}_{2} = \int_{0}^{\infty} P(0)F^{*}(t)dt +  \frac{A}{2} \int_{0}^{\infty} |(F^{*})'(t)|dt$.
We will bound the Lagrangian dual objective. 

By Lemma~\ref{lem:general-energy} (second statement), 
for every job $j$ such that $s^{*}(t) = 0$ for some $t \in [r_{j},d_{j}]$, 
$\lambda_{j} = \frac{P(s^{c})}{s^{c}}$. By the definition of the critical speed,
$\lambda_{j} \leq \frac{{P}(z)}{z}$ for any $z > 0$. Therefore,
\begin{equation} 	\label{eq:general-energy}
\sum_{j}  \int_{r_{j}}^{d_{j}} s_{j}(t) F(t) \biggl( \lambda_{j} - \frac{{P}(s(t))}{s(t)}\biggl)dt \leq 0
\end{equation}
where in the sum is taken over jobs $j$ such that 
$s^{*}(t) = 0$ for some $t \in [r_{j},d_{j}]$.  Therefore,
\begin{align*}
L_{1}(s,\lambda) &:= \sum_{j} \lambda_{j}p_{j} 
	- \sum_{j}  \int_{r_{j}}^{d_{j}} s_{j}(t) F(t) \biggl( \lambda_{j} - \frac{{P}(s(t))}{s(t)}\biggl) \one_{\{s(t) > 0\}}\one_{\{F(t) = 1\}}dt \\
	&\geq \beta E^{*}_{1} - \max_{s,F} \sum_{j} \int_{r_{j}}^{d_{j}} s_{j}(t) F(t) 
		\biggl[ \beta P'(s^{*}(t)) - \frac{{P}(s(t))}{s(t)}\biggl] \one_{\{s(t) > 0\}} \one_{\{F(t) = 1\}} \one_{\{s^{*}(t) > 0\}} dt\\
	&\geq \beta E^{*}_{1} - \max_{s} \int_{0}^{\infty} s(t)  
		\biggl[ \beta P'(s^{*}(t)) - \frac{{P}(s(t))}{s(t)}\biggl]  \one_{\{s(t) > 0\}}\one_{\{F(t) = 1\}} \one_{\{s^{*}(t) > 0\}} dt \\
	&\geq \beta E^{*}_{1} - \int_{0}^{\infty}  
		\biggl[ \beta P'(s^{*}(t))  \bar{s}(t) - {P}(\bar{s}(t)) \biggl] \one_{\{s(t) > 0\}}\one_{\{F(t) = 1\}}\one_{\{s^{*}(t) > 0\}} dt \\
	&\geq \beta E^{*}_{1} - \frac{1}{2} \int_{0}^{\infty}  
		\biggl[ \beta P'(s^{*}(t))  \bar{s}(t) - {P}(\bar{s}(t)) \biggl] \one_{\{s^{*}(t) > 0\}} dt \\
		& \qquad - \frac{1}{2} \int_{0}^{\infty}  
		\biggl[ \beta P'(s^{*}(t))  \bar{s}(t) - {P}(\bar{s}(t)) \biggl] \one_{\{F(t) = 1\}} dt 
\end{align*}
where in the second line, the sum is taken over jobs $j$ such that 
$s^{*}(t) > 0$ for all $t \in [r_{j},d_{j}]$.
The first inequality follows (\ref{eq:general-energy}) and Lemma~\ref{lem:general-energy} (first statement).
The second inequality holds since $F(t) \leq 1$ and $s(t) = \sum_{j} s_{j}(t)$. 
The third inequality is due to the first order derivative 
and $\bar{s}(t)$ is the solution of equation $P'(z(t)) = \beta P'(s^{*}(t))$.
In fact $\bar{s}(t)$ maximizes function $s(t)\beta P'(s^{*}(t)) - {P}(s(t))$.

As the energy power function $P(z) = z^{\alpha} + g$ where $\alpha \geq 1$ and $g \geq 0$, 
$\bar{s}(t)^{\alpha-1} = \beta (s^{*}(t))^{\alpha-1}$. Therefore,
\begin{align*}
L_{1}(s,\lambda)  &\geq \beta E^{*}_{1} - \frac{1}{2}\int_{0}^{\infty}  
		\biggl( \beta \alpha (s^{*}(t))^{\alpha-1}  \bar{s}(t) - (\bar{s}(t))^{\alpha} - g \biggl) \one_{\{s^{*}(t) > 0\}} dt \\
		& \qquad - \frac{1}{2}\int_{0}^{\infty}  
		\biggl( \beta \alpha (s^{*}(t))^{\alpha-1}  \bar{s}(t) - (\bar{s}(t))^{\alpha} - g \biggl) \one_{\{F(t) =1\}} dt \\
& = \beta E^{*}_{1} - \int_{0}^{\infty} (\alpha-1) \beta^{\alpha/(\alpha-1)}(s^{*}(t))^{\alpha}dt 
		+ \frac{1}{2} \int_{0}^{\infty} g \one_{\{s^{*}(t) > 0\}} dt 
		+ \frac{1}{2} \int_{0}^{\infty} g \one_{\{F(t) =1\}} dt \\
&= \biggl[ \beta - (\alpha-1) \beta^{\alpha/(\alpha-1)} \biggl] E^{*}_{1} 
	+ \frac{1}{2} \int_{0}^{\infty} g \one_{\{s^{*}(t) > 0\}} dt 
		+ \frac{1}{2} \int_{0}^{\infty} g \one_{\{F(t) =1\}} dt
\end{align*}
Choose $\beta = 1/\alpha^{\alpha-1}$, we have that
\begin{align*}
L(s,F,\lambda) \geq \frac{1}{\alpha^{\alpha}} E^{*}_{1} 
	+ \frac{1}{2} \int_{0}^{\infty} g \one_{\{s^{*}(t) > 0\}} dt 
		+ \frac{1}{2} \int_{0}^{\infty} g \one_{\{F(t) =1\}} dt
	+ \frac{A}{2} \int_{0}^{\infty} |F'(t)|dt
\end{align*}

In the following, we claim that 
$$L_{2}(F) := \frac{1}{2} \int_{0}^{\infty} g \one_{\{s^{*}(t) > 0\}} dt 
		+ \frac{1}{2} \int_{0}^{\infty} g \one_{\{F(t) =1\}} dt + \frac{A}{2} \int_{0}^{\infty} |F'(t)|dt
\geq  E^{*}_{2}/4
$$
for any feasible solution $(s,F)$ of the relaxation.

Consider the algorithm schedule. An \emph{end-time} $u$ is a moment in the schedule such
that the machine switches from the idle state to the sleep state. 
Conventionally, the first end-time in the schedule is 0. 
Partition the time line into phases. A \emph{phase} $[u,v)$ is a time interval such that $u,v$
are consecutive end-times. Observe that in a phase, the schedule
has transition cost $A$ and there is always a new job released in a phase
(otherwise the machines would not switch to non-sleep state).
We will prove the claim on every phase. In the following, 
we are interested in phase $[u,v)$ and whenever we mention $L_{2}(F)$, 
it refers to $\frac{1}{2} \int_{u}^{v} g \one_{\{s^{*}(t) > 0\}} dt 
		+ \frac{1}{2} \int_{u}^{v} g \one_{\{F(t) =1\}} dt + \frac{A}{2} \int_{u}^{v} |F'(t)|dt$. 

By the algorithm, the static energy of the schedule during 
the idle time is $A$, \linebreak i.e., $\int_{u}^{v}  g\one_{\{s^{*}(t) = 0\}} dt = A$. 
Let $(s,F)$ be an arbitrary feasible of solution of the relaxation. 

If during $[u,v)$, the machine following solution $(s,F)$ makes a transition from
non-sleep state to sleep state or inversely then 
$L_{2}(F) \geq \frac{1}{2} \int_{u}^{v} g \one_{\{s^{*}(t) > 0\}} dt + \frac{A}{2}$. Hence
$$
L_{2}(F) \geq \frac{1}{4} \biggl( \int_{u}^{v} g \one_{\{s^{*}(t) > 0\}} dt 
	      	+  \int_{u}^{v}  g \one_{\{s^{*}(t) = 0\}}dt + A \biggl) = \frac{1}{4}\restr{E^{*}_{2}}{[u,v)}.
$$

If during $[u,v)$, the machine following solution $(s,F)$ makes no transition (from
non-sleep static to sleep state or inversely) then $F(t) = 1$ during $[u,v)$ in order to 
process jobs released in the phase. Therefore, 
\begin{align*}
L_{2}(F) &\geq \frac{1}{2} \int_{u}^{v} g \one_{\{s^{*}(t) > 0\}} dt 
		+ \frac{1}{2} \int_{u}^{v} g \one_{\{F(t) =1\}} dt 
		=  \frac{1}{2} \int_{u}^{v} g \one_{\{s^{*}(t) > 0\}} dt 
		+ \frac{1}{2} \int_{u}^{v} g dt \\
	      & \geq \frac{1}{2} \int_{u}^{v} g \one_{\{s^{*}(t) > 0\}} dt 
	      	+ \frac{1}{4} \int_{u}^{v}  g \one_{\{s^{*}(t) = 0\}}dt + \frac{A}{4} \\
	      & \geq \frac{1}{4} \biggl( \int_{u}^{v} g \one_{\{s^{*}(t) > 0\}} dt 
	      	+  \int_{u}^{v}  g \one_{\{s^{*}(t) = 0\}}dt + A \biggl) = \frac{1}{4}\restr{E^{*}_{2}}{[u,v)}
\end{align*}
where the second inequality follows the algorithm: as the machine switches to sleep state 
at time $v$, it means that the total idle duration in $[u,v)$ incurs a cost $A$.

In conclusion, the dual $L(s,F,\lambda) \geq E^{*}_{1}/\alpha^{\alpha} + E^{*}_{2}/4$
whereas the primal is $E^{*}_{1}+ E^{*}_{2}$. Thus, 
the competitive ratio is at most $\max\{4,\alpha^{\alpha}\}$.
\end{proof}

\section{Minimizing Energy plus Weighted Flow-Time in Speed Scaling with 
Power Down Model}		\label{sec:4S-energy+flow}

\paragraph{The problem.} 

We consider the problem of minimizing the total weighted flow-time
plus energy on a single in the general energy model. 
Again, the machine has a transition cost $A$ from sleep state to active state.
The power energy consumption of the machine at time $t$ in its active state is 
$P(s(t)) = s(t)^{\alpha} + g$
where $\alpha \geq 1$ and $g \geq 0$ and $s(t)$ is the machine speed at time $t$. 
Recall that the \emph{dynamic energy} is 
$\int_{0}^{\infty} s^{\alpha}(t)dt$ and the \emph{static energy} is
$\int_{0}^{\infty}g dt$ (where the integrals are taken during
active periods).
Jobs arrive over time, a job $j$ is released at time $r_{j}$, has weight $w_{j}$ 
and requires $p_{j}$ units of processing volume if it is processed 
on machine $i$. A job could be processed preemptively, i.e., a job could be interrupted and resumed later. 
The \emph{flow-time} of a job $j$ is $C_{j} - r_{j}$ where
$C_{j}$ is the completion time of the job.
At any time, the scheduler has to determine the state and the speed of every machine (it it is active) 
and also a policy how to execute jobs.
The objective is to minimize the total weighted flow-time of all jobs plus the total 
energy (including the wake-up cost). 

\paragraph{Formulation.} Similar as the previous section, we 
make use of the properties of Heaviside step function and Dirac delta function to encode 
the machine states and the transition cost. Let $F(t)$ be a function
indicating whether the machine $i$ is in active state at time $t$, i.e., $F(t) = 1$ if 
the machine is active at $t$ and $F(t) = 0$ if it is in the sleep state. 
Assume that the machine initially is in the sleep state. Then
$A \int_{0}^{+\infty} |F'(t)|dt$ equals twice the transition cost
of the machine.
Let $s_{j}(t)$ be the variable that represents 
the speed of job $j$ at time $t$. Let $C_{j}$ be 
a variable representing the completion time of $j$. 
The problem could be relaxed as the following (non-convex) program.

  \begin{alignat}{3}		\label{relaxation}
    \text{minimize } \int_{0}^{\infty} 2 P\biggl( \sum_{j} s_{j}(t)\biggl) & F(t) dt 
    			+ 2\sum_{j} \biggl(\int_{r_{j}}^{C_{j}} s_{j}(t)F(t)dt\biggl) \frac{w_{j}}{p_{j}}(C_{j} - r_{j}) \notag \\
			&\qquad \quad + A \int_{0}^{\infty} |F'(t)|dt \\
	\text{subject to} \qquad    \int_{r_{j}}^{C_{j}} s_{j}(t)F(t)dt &= p_{j}  \qquad \forall j \notag \\
	      s_{j}(t) &\geq 0 \qquad \forall j, t \geq r_{j} \notag \\
		F(t) &\in \{0,1\}  \qquad \forall  t.		\notag
  \end{alignat}
The first constraints ensure that every job $j$ must 
be completed by time $C_{j}$. 
In the objective function, the first and second terms represent
twice the consumed energy and the total weighted flow-time, respectively. 
Note that in the second term, $\int_{r_{j}}^{C_{j}} s_{j}(t)F(t)dt = p_{j}$ by the constraints. 
The last term stands for twice the transition cost.


\paragraph{Preliminaries.}
We say that a job $j$ is \emph{pending} at time $t$ if it is not completed, i.e., $r_{j} \leq t < C_{j}$. 
At time $t$, denote $q_{j}(t)$ the \emph{remaining} processing volume of job $j$.
The \emph{total weight} of pending jobs at time $t$ 
is denoted as $W(t)$. 
The \emph{density} of a job $j$ is $\delta_{j} = w_{j}/p_{j}$. 
Define the \emph{critical speed} $s^{c}$ of the machine 
as $\arg \min_{z \geq 0} P(z)/z$. As $P(z) = z^{\alpha} + g$, by 
the first order condition, $s^{c}$ satisfies $(\alpha-1)(s^{c})^{\alpha} = g$.

\subsection{The Algorithm.} 

We first describe the algorithm informally. In the speed scaling model, all previous algorithms 
explicitly or implicitly balance the weighted flow-time of jobs and the consumed energy to process such jobs. 
That could be done by setting the machine speed at any time $t$ proportional to 
some function of the total weight of pending jobs (precisely, proportional to 
$W(t)^{1/\alpha}$ where $W(t)$ is the total weight of pending jobs).
Our algorithm follows the same idea of balancing.
However, in the general energy model, the algorithm would not be competitive if 
the speed is always set proportionally to $W(t)^{1/\alpha}$ since the static energy might be large
due to the long active periods of the machine. Hence, even if the total weight of 
pending jobs on the machine is small, in some situation the speed is maintained larger than 
$W(t)^{1/\alpha}$. In fact, it will be set to be the critical 
speed $s^{c}$, defined as $\arg \min P(z)/z$. 

An issue while dealing with the general model is to determine when a machine is waken up.
Again, if the total weight of pending jobs is small and the machine is active,
then the static energy is large. Otherwise if pending jobs remain for long time then the 
weight flow-time is large. The algorithm also balances the costs by making a plan and 
switching the machine into active state at appropriate moments. 
If new job is released then the plan, together with its starting time, will be changed.  

\paragraph{Description of algorithm.}
At any time $t$, the machine maintains the following policy
in different states: the \emph{working state} (the machine is active and currently processes
some job), the \emph{idle state} (the machine is active but currently processes no job) and 
the \emph{sleep state}.

\begin{description}
	\item[In working state.] If $\frac{\alpha}{\alpha-1}W(t)^{(\alpha-1)/\alpha} > P(s^{c})/s^{c}$
		then the machine speed is set as $ W(t)^{1/\alpha}$.
		Otherwise, the speed is set as $ s^{c}$.
		At any time, the machine processes the highest  density job
		among the pending ones. 
	\item[In idle state.]
		\begin{enumerate}
			\item If $\frac{\alpha}{\alpha-1}W(t)^{(\alpha-1)/\alpha} > P(s^{c})/s^{c}$
			then switch to the working state.
			\item If $0 < \frac{\alpha}{\alpha-1}W(t)^{\frac{\alpha-1}{\alpha}} \leq P(s^{c})/s^{c}$
			then make a plan to process the pending jobs with speed (exactly) $s^{c}$
			in non-increasing order of their  density. So the plan consists of a single 
			block (with no idle time) and the block length could be explicitly computed (given 
			the processing volumes of all jobs and speed $ s^{c}$). Hence, the total energy consumed in the plan 
			could also be computed and it is independent of the starting time of the plan. 
		
			Choose the starting time of the plan in such a way that the total energy consumed 
			in the plan equals the total weighted flow-time of all jobs in the plan. 
			There always exists
			such starting time since if the plan begins immediately at $t$, the energy is larger than  
			the weighted flow-time; 
			and inversely if the starting time is large enough, the latter dominates the former.  
		
			At the starting time of a plan, switch to the working state. (Note that the plan together with 
			its starting time could be changed due to the arrival of new jobs.)
			\item Otherwise, if the total duration of idle state from the last wake-up equals
				$A/g$ then switch to sleep state.
		\end{enumerate}
	\item[In sleep state.] Use the same policy as
		the first two steps of the idle state.   
\end{description}
  
\subsection{Analysis} 
  
The Lagrangian dual of program (\ref{relaxation}) is $\max \min_{x,s,C,F} L$ 
where $L$ is the corresponding Lagrangian function where the maximum is taken over 
dual variables. The purpose of the section is to choose appropriate dual variables
and prove that for any feasible solution $(x,s,C,F)$ of the primal, the Lagragian dual
is bounded by a desired factor from the primal.

\paragraph{Dual variables.} 
Denote the dual variables corresponding to the first constraints of (\ref{relaxation}) 
as $\lambda_{j}$'s. Set all dual variables (corresponding to the primal (\ref{relaxation})) 
except $\lambda_{j}$'s equal to 0.
The values of dual variables $\lambda_{j}$'s is defined as the follows. 

Fix a job $j$. At the arrival of a job $j$, 
rename pending jobs as $\{1,\ldots,k\}$
in non-increasing order of their  densities, i.e., 
$p_{1}/w_{1} \leq \ldots \leq p_{k}/w_{k}$ 
(note that $p_{a}/w_{a}$ is the inverse of job $a$'s density).
Denote $W_{a} = w_{a} + \ldots + w_{k}$ for $1 \leq a \leq k$.

Define $\lambda_{j}$ such that 
\begin{equation}		\label{eq:lambda}
\lambda_{j}p_{j} = w_{j} \sum_{a = 1}^{j} \frac{q_{a}(r_{j})}{W_{a}^{1/\alpha}}	
		+ W_{j+1} \frac{q_{j}(r_{j})}{W_{j}^{1/\alpha}} + P(s^{c}) \frac{q_{j}(r_{j})}{s^{c}}
\end{equation}
Note that $q_{j}(r_{j}) = p_{j}$.
If job $j$ is processed with speed larger than $s^{c}$ then
the first term stands for the weighted flow-time of $j$ and
the second term represents an upper bound of the increase 
in the weighted flow-time of jobs with density smaller than $\delta_{j}$. 
Observe that due to arrival of $j$, the jobs with higher density than $\delta_{j}$
are completed earlier and the ones with smaller density than $\delta_{j}$
may have higher flow-time. 
Informally, the second sum in (\ref{eq:lambda}) captures 
the marginal change in the total weighted flow-time. 
The third term in (\ref{eq:lambda}) is introduced in order to cover energy consumed during the 
execution periods of job $j$ if it is processed by speed $s^{c}$. That term is necessary 
since during such periods the 
energy consumption and the weighted flow-time is not balanced.

\paragraph{}
The Lagrangian function $L(x,s,C,F,\lambda)$ with the chosen dual variables becomes

\begin{align*}
A \int_{0}^{\infty}& |F'(t)|dt  + 2\int_{0}^{\infty} P\biggl( \sum_{j} s_{j}(t)\biggl)F(t)dt
    		+ 2 \sum_{j} \delta_{j}(C_{j} - r_{j}) \int_{r_{j}}^{C_{j}} s_{j}(t)F(t)dt \\
		&+ \sum_{j} \lambda_{j} \biggl( p_{j} - \int_{r_{j}}^{C_{j}} s_{j}(t)F(t)dt \biggl) \\
= \sum_{j} & \lambda_{j}p_{j}  +  A \int_{0}^{\infty} |F'(t)|dt 
			+ \sum_{j} \int_{r_{j}}^{C_{j}}  \delta_{j}(C_{j} - r_{j}) s_{j}(t)F(t)dt\\
		&- \sum_{j} \int_{r_{j}}^{C_{j}} s_{j}(t)F(t)
			\biggl( \lambda_{j} - 2\frac{P(s(t))}{s(t)} 
				- \delta_{j}(C_{j} - r_{j}) \biggl) dt	
\end{align*}

\paragraph{Notations.}
We denote $s^{*}(t)$ the machine speed at time $t$ by the algorithm. 
So by the algorithm, if $s^{*}(t) > 0$ then $s^{*}(t) \geq s^{c}$.
Let $\mathcal{E}^{*}_{1}$ and $\mathcal{E}^{*}_{2}$ be the total dynamic and static 
energy consumed by the algorithm schedule, respectively. 
In other words, $\mathcal{E}^{*}_{1} = \int_{0}^{\infty} (s^{*}(t))^{\alpha} dt$
and $\mathcal{E}^{*}_{2} = \int_{0}^{\infty} g$ where the integral is taken over all moments 
$t$ where the machine is active (either in working or in idle states).
Additionally, let $\mathcal{E}^{*}_{3}$ be the total transition cost of the machine.
Moreover, let $\mathcal{F}^{*}$ be the total weighted flow-time of all jobs 
in the schedule. 

We relate the cost of the schedule (due to the algorithm) and the chosen values of dual variables by the 
following lemma. Note that by definition of $\lambda_{j}$'s, we have that 
$\sum_{j} \lambda_{j}p_{j} \geq \mathcal{F}^{*}$.

\begin{lemma}		\label{lem:general-energy-lambda}
It holds that
$
2\mathcal{E}^{*}_{1} + 3\mathcal{E}^{*}_{2}
	\geq  \mathcal{F}^{*} 
$
and
$
\sum_{j} \lambda_{j}p_{j} \geq \mathcal{E}^{*}_{1} 
$.
\end{lemma}
\begin{proof}
We prove the first inequality.
Consider times $t$ where the machine speed is $ s^{c}$.   
By the algorithm $P(s^{c})/s^{c} \geq \frac{\alpha}{\alpha-1}W(t)^{(\alpha-1)/\alpha}$.
Recall that by the definition of critical speed $g = (\alpha-1)(s^{c})^{\alpha}$,
so $\alpha (s^{c})^{\alpha} = P(s^{c})$. Therefore,
$s^{c} \geq (\alpha-1)^{-1/(\alpha-1)} W(t)^{1/\alpha}$. Hence,
\begin{align*}
2 \mathcal{E}^{*}_{2} \geq (\alpha - 1)^{\frac{1}{\alpha-1}} \int_{0}^{\infty} g \one_{\{s^{*}(t) = s^{c}\}} dt 
	 &=(\alpha - 1)^{\frac{1}{\alpha-1}} 
		\int_{0}^{\infty} (\alpha - 1) (s^{c})^{\alpha} \one_{\{s^{*}(t) = s^{c}\}}dt \\
	&\geq \int_{0}^{\infty} W(t) \one_{\{s^{*}(t) = s^{c}\}}dt
\end{align*}

Now consider times $t$ where the machine speed is $ W(t)^{1/\alpha}$
strictly larger than $ s^{c}$. Thus
the dynamic energy consumed on such periods is
$$
\mathcal{E}^{*}_{1} \geq \int_{0}^{\infty}(s^{*}(t))^{\alpha} \one_{\{s^{*}(t) > s^{c}\}}dt 
\geq \int_{0}^{\infty}  W(t) \one_{\{s^{*}(t) > s^{c}\}}dt 
$$

For periods where $s^{*}(t) = 0$ while some jobs are still pending
on the machine, by the algorithm plan, the total weighted flow time of jobs in such periods 
is bounded by $(\mathcal{E}^{*}_{1} + \mathcal{E}^{*}_{2})$
Therefore,
$$
2\mathcal{E}^{*}_{1} + 3\mathcal{E}^{*}_{2}
	\geq   \int_{0}^{\infty} W(t) dt
	=  \mathcal{F}^{*}. 
$$

In the rest, we prove the second inequality $\sum_{j} \lambda_{j}p_{j} \geq \mathcal{E}^{*}_{1}$.
By the definition of $\lambda_{j}$'s (particularly the third term in (\ref{eq:lambda})),
$\sum_{j} \lambda_{j}p_{j}$ covers the total energy of machine during all intervals where 
the machine processes jobs by speed $s^{c}$. Denote $\Gamma$ as 
$\sum_{j} \lambda_{j}p_{j}$ subtracting the energy incurred during periods where the machine 
speed is $s^{c}$. We need to prove that $\Gamma$  is enough to cover the total energy 
incurred over all moments where the machine speed is strictly larger than $s^{c}$.  
In the following, we are interested only in such moments.

Consider a job $k$ processed at time $t$ with speed larger than $s^{c}$. 
By the definition of $\lambda_{j}$'s, 
$\Gamma$ contributes to time $t$ an amount at least $\sum_{j} w_{j}/ W(t)^{1/\alpha}$ where
the sum is taken over pending jobs $j$ with smaller density than that of $k$.
The latter is exactly $W(t)$. Thus, $\Gamma$ contributes to time $t$
an amount $W(t)^{(\alpha-1)/\alpha}$.

Now consider an arbitrarily small interval $[a,b]$ where the machine processes only job $k$
and the speed is strictly larger than $s^{c}$.
Let $W$ be the total weight of pending jobs over $[a,b]$.  The processing 
amount of $k$ done over $[a,b]$ is $ W^{1/\alpha}(b-a)$ while the 
energy amount consumed in that interval is $ W(b-a)$. Hence, 
in average the machine spends $W^{(\alpha-1)/\alpha}$ 
(dynamic) energy unit at time $t$.     

Therefore, during periods where the machine speed is larger than $s^{c}$,
$\Gamma$ is increase at rate proportionally to the one of the dynamic energy.
The second inequality of the lemma follows.
\end{proof}

\begin{corollary}		\label{cor:general-energy-lambda}
It holds that
$
 	\sum_{j} \lambda_{j}p_{j}
		\geq\frac{7}{8}\mathcal{E}^{*}_{1} + 
			 \frac{1}{16}\mathcal{F}^{*} - \frac{3}{16} \mathcal{E}^{*}_{2}
$.
\end{corollary}
\begin{proof}
By the previous lemma, we deduce that
\begin{align*}
	\sum_{j} \lambda_{j}p_{j}
		\geq \mathcal{E}^{*}_{1} 
		\geq \frac{7}{8}\mathcal{E}^{*}_{1} + 
			\frac{1}{8} \biggl(
					\frac{1}{2}\mathcal{F}^{*} - \frac{3}{2} \mathcal{E}^{*}_{2} \biggl)
		=  \frac{7}{8}\mathcal{E}^{*}_{1} + 
			 \frac{1}{16}\mathcal{F}^{*} - \frac{3}{16} \mathcal{E}^{*}_{2}.
\end{align*}
\end{proof}

In the following, we show the main technical lemma.

\begin{lemma}		\label{lem:energy-general-main}
Let $j$ be an arbitrary job.  
Then, for every $t \geq r_{j}$
\begin{align}		\label{eq:energy-general-main}
\lambda_{j} - \delta_{j} (t-r_{j}) \leq \max \biggl\{ \frac{\alpha}{\alpha-1} W(t)^{\frac{\alpha-1}{\alpha}} 
		+ \frac{P(s^{c})}{s^{c}}, 2\frac{P(s^{c})}{s^{c}} \biggl \} 
\end{align}
\end{lemma} 
\begin{proof}
Fix a job $j$. 
%
We prove by induction on the number of released jobs
after $r_{j}$. The base case follows Lemma \ref{lem:base}
and the induction step is done by Lemma \ref{lem:induction}. 
\end{proof}

\begin{lemma}		\label{lem:base}
If no new job is released after $r_{j}$ then inequality (\ref{eq:energy-general-main}) holds.
\end{lemma}
\begin{proof}
Denote the instance as $\mathcal{I}_{0}$.
At $r_{j}$, rename jobs in non-increasing order of their densities, i.e., 
$p_{1}/w_{1} \leq \ldots \leq p_{n}/w_{n}$ (note that $p_{a}/w_{a}$ is the inverse of job $a$'s density).
Denote $W_{a} = w_{a} + \ldots + w_{n}$ for $1 \leq a \leq n$.

By definition of $\lambda_{j}$,
\begin{align*}		
\lambda_{j} - \frac{P(s^{c})}{s^{c}} 
	= \delta_{j}\biggl[ \frac{q_{1}(r_{j})}{ W_{1}^{1/\alpha}} + \ldots + \frac{q_{j}(r_{j})}{ W_{j}^{1/\alpha}} \biggl]
		+ \frac{W_{j+1}}{W_{j}^{1/\alpha}}
\end{align*}

Let $C_{a}(\mathcal{I}_{0})$ be the completion time of job $a$ for every $a$. Moreover,
let $\ell$ be the largest job index such that 
$\frac{\alpha}{\alpha-1}W_{\ell}^{(\alpha-1)/\alpha} \geq P(s^{c})/s^{c}$.
In other words, job $\ell$ is processed by speed strictly larger than $s^{c}$ and the other jobs
with larger index (if exist) will be processed by speed $s^{c}$.
Fix a time $t$, let $k$ be the pending job at $t$ with the smallest index.
We prove first the following claim.

\begin{claim}
It holds that
$$
\lambda_{j} - \delta_{j} (t-r_{j}) - \frac{P(s^{c})}{s^{c}}
\leq \max \biggl\{ \frac{w_{k}}{W_{k}^{1/\alpha}} + 
		\frac{w_{k+1}}{W_{k+1}^{1/\alpha}} + \ldots + \frac{w_{n}}{W_{n}^{1/\alpha}}, \frac{P(s^{c})}{s^{c}} \biggl\}
$$
\end{claim}
\begin{claimproof}
We consider different cases.
\paragraph{Case 1: $\ell \leq j$.} In this case, job $j$ will be processed by speed $s^{c}$.
	\begin{description}
		\item[Subcase 1.1: $t \leq C_{\ell}(\mathcal{I}_{0})$.]
		During interval $[r_{j},t]$, the machine has completed jobs $1, \ldots, k-1$ and has processed a part of job $k$.
		Precisely, during $[r_{j},t]$ the machine has processed $q_{a}(r_{j})$ units of job $a$
		for every job $1 \leq a < k$ and has executed $(q_{k}(r_{j}) - q_{k}(t))$ units of job $k$.   
		Moreover, every job $1 \leq a \leq k$ is processed with speed $W_{a}^{1/\alpha}$.  
		Therefore,  
\begin{align*}
\lambda_{j} - \delta_{j} (t-r_{j}) - \frac{P(s^{c})}{s^{c}}
&= \delta_{j} \biggl[ \frac{q_{k}(t)}{W_{k}^{1/\alpha}} + 
		\frac{q_{k+1}(r_{j})}{W_{k+1}^{1/\alpha}} + \ldots + \frac{q_{j}(r_{j})}{W_{j}^{1/\alpha}} \biggl]
		+ \frac{W_{j+1}}{W_{j}^{1/\alpha}} \\
&\leq \delta_{j} \biggl[ \frac{p_{k}}{W_{k}^{1/\alpha}} + 
		\frac{p_{k+1}}{W_{k+1}^{1/\alpha}} + \ldots + \frac{p_{j}}{W_{j}^{1/\alpha}} \biggl]
		+ \frac{W_{j+1}}{W_{j}^{1/\alpha}} \\
&= \delta_{j} \biggl(  \frac{w_{k}}{\delta_{k} W_{k}^{1/\alpha}} + 
		\frac{w_{k+1}}{\delta_{k+1} W_{k+1}^{1/\alpha}} + \ldots + \frac{w_{j}}{\delta_{j} W_{j}^{1/\alpha}} \biggl)
		+ \frac{W_{j+1}}{ W_{j}^{1/\alpha}} \\
&\leq \frac{w_{k}}{W_{k}^{1/\alpha}} + 
		\frac{w_{k+1}}{W_{k+1}^{1/\alpha}} + \ldots + \frac{w_{j}}{W_{j}^{1/\alpha}}
		+ \frac{w_{j+1}}{W_{j+1}^{1/\alpha}} + \ldots + \frac{w_{n}}{W_{n}^{1/\alpha}} \\
&\leq \int_{0}^{W(t)} \frac{dz}{z^{1/\alpha}}  
= \frac{\alpha}{\alpha - 1} W(t)^{\frac{\alpha-1}{\alpha}} 	
\end{align*}
		The first inequality is because $q_{a}(r_{j}) \leq p_{a}$ for every job $a$.
		The first equality is due to the definition of the density.
		The second inequality follows since $\delta_{j} \leq \delta_{a}$ for every job $a \leq j$
		and $W_{j+1} \geq \ldots \geq W_{n}$.
		The third inequality holds since function $z^{-1/\alpha}$ is decreasing. 

		\item[Subcase 1.2: $t > C_{\ell}(\mathcal{I}_{0})$.]
		In this case $k > \ell$, i.e., during $[r_{j},t]$ the machine $i$ has completed 
		jobs $1, \ldots, \ell$. Similarly as the previous subcase, we have 
\begin{align*}
\lambda_{j} &- \delta_{j} (t-r_{j}) -  \frac{P(s^{c})}{s^{c}}
	\leq   \sum_{a=\ell+1}^{n} \frac{w_{a}}{W_{a}^{1/\alpha}} \\
	&\leq \int_{0}^{W_{\ell+1}} \frac{dz}{z^{1/\alpha}} 
		= \frac{\alpha}{\alpha - 1} W_{\ell+1}^{\frac{\alpha-1}{\alpha}} 	\leq \frac{P(s^{c})}{s^{c}} 
\end{align*}
where the last inequality follows the definition of $\ell$.
	\end{description}

\paragraph{Case 2: $\ell > j$.} In this case, job $j$ will be processed with speed strictly larger than $s^{c}$.
	\begin{description}
		\item[Subcase 2.1: $t \leq C_{j}(\mathcal{I}_{0})$.] The proof is done in the same manner as in Subcase 1.1.
		\item[Subcase 2.2: $t > C_{j}(\mathcal{I}_{0})$.] 
		For simplicity, denote $C_{a} = C_{j}(\mathcal{I}_{0})$.
		Partition time after $C_{j}(\mathcal{I}_{0})$ as $\cup_{a = j}^{n} [C_{a}, C_{a+1})$.  
		During an interval $ [C_{a}, C_{a+1})$, the weight $W_{a}$ is unchanged so 
		to show inequality (\ref{eq:energy-general-main}), it is sufficient to prove 
		at $t = C_{j}, C_{j+1}, \ldots, C_{n-1}$. 
		
		We prove again by induction. For the base case $t = C_{j}$, the claim inequality holds
		by the previous case. Assume that the inequality holds at $t = C_{a}$, we will prove that 
		it holds at $t = C_{a+1}$ for $a \geq j$. We are interested only in $\tau \in [C_{a},C_{a+1})$.
		Let $V(\tau) = w_{a} q_{a}(\tau)/p_{a} + w_{a+1} + \ldots + w_{n}$. Informally, $V(\tau)$ is the fractional
		weight of pending jobs at time $\tau$. 
		
		During period $[\tau,\tau+d\tau]$ assume that the total fractional weight of pending jobs 
		varies by $dV(\tau)$. During the same period $[\tau,\tau+d\tau]$,
		the total processing volume done by algorithm is at least
		$W_{a}^{1/\alpha}d\tau$ since the speed is either $W(\tau)^{1/\alpha} (= W_{a}^{1/\alpha})$ or 
		$s^{c}$ but in the latter, by the algorithm, $s^{c} \geq W(\tau)^{1/\alpha}$.
		Moreover, jobs $a$ processed during $[\tau,\tau+d\tau]$ have density at most $\delta_{j}$.
		Therefore, $dV(\tau) \leq  \delta_{j} W_{a}^{1/\alpha}d\tau$. In other words,
		$V'(\tau)d\tau \leq \delta_{j}W_{a}^{1/\alpha}d\tau$. Taking integral, we get   
		\begin{align}	\label{eq:weight-a}
		w_{a} = W_{a} - W_{a+1} = \int_{C_{a}}^{C_{a+1}} V'(\tau)d\tau 
			\leq \int_{C_{a}}^{C_{a+1}} \delta_{j} W_{a}^{1/\alpha}d\tau
			= \delta_{j} W_{a}^{1/\alpha} (C_{a+1} - C_{a})
		\end{align}
		Therefore, 
		\begin{align*}
			\lambda_{j} - \delta_{j} (C_{a+1}-r_{j}) - \frac{P(s^{c})}{s^{c}}
				&= \lambda_{j} - \delta_{j} (C_{a}-r_{j}) - \frac{P(s^{c})}{s^{c}} -  \delta_{j} (C_{a+1}-C_{a}) \\
				&\leq \max \biggl \{ \frac{w_{a}}{W_{a}^{1/\alpha}} + \ldots + \frac{w_{n}}{W_{n}^{1/\alpha}} 
								,\frac{P(s^{c})}{s^{c}} \biggl \} 
						- \frac{w_{a}}{W_{a}^{1/\alpha}} \\
				&\leq \max \biggl \{ \frac{w_{a+1}}{W_{a+1}^{1/\alpha}} + \ldots + \frac{w_{n}}{W_{n}^{1/\alpha}}
								,\frac{P(s^{c})}{s^{c}} \biggl \} 
		\end{align*}
		where the first inequality is due to the induction hypothesis and inequality (\ref{eq:weight-a}).	
	\end{description}
Combining all the cases, the claim holds.
\end{claimproof}

Using the claim, the lemma follows immediately as shown below.
\begin{align*}
\lambda_{j} &- \delta_{j} (t-r_{j}) - \frac{P(s^{c})}{s^{c}}
\leq \max \biggl \{ \frac{w_{k}}{W_{k}^{1/\alpha}} + \ldots + \frac{w_{n}}{W_{n}^{1/\alpha}}, \frac{P(s^{c})}{s^{c}} \biggl \} \\
&\leq \max \biggl \{ \int_{0}^{W(t)} \frac{dz}{z^{1/\alpha}}, \frac{P(s^{c})}{s^{c}} \biggl \}   
= \max \biggl \{ \frac{\alpha}{\alpha - 1} W(t)^{\frac{\alpha-1}{\alpha}}, \frac{P(s^{c})}{s^{c}} \biggl \}  	
\end{align*}
where the inequality holds since function $z^{-1/\alpha}$ is decreasing. (Recall that $k$ is the pending job 
at time $t$ with the smallest index.)
\end{proof}

\begin{lemma}		\label{lem:induction}
Assume that inequality (\ref{eq:energy-general-main}) holds if 
there are $(n-1)$ jobs released after $r_{j}$. Then the inequality also holds
if $n$ jobs are released after $r_{j}$.
\end{lemma}
\begin{proof}
Denote the instance as $\mathcal{I}_{n}$.
Among such jobs, let $n$ be the last released one (at time $r_{n}$).
By induction hypothesis, it remains to prove the lemma inequality for $t \geq r_{n}$.

We first show the claim that inequality (\ref{eq:energy-general-main}) holds for any time 
$t \geq C_{j}(\mathcal{I}_{n})$ by a similar argument as in Subcase 2.2 of the previous claim. 
Indeed, we prove the claim by fixing the 
processing volume of job $n$ and varying its weight $w_{n}$. Note that 
$C_{j}(\mathcal{I}_{n})$ depends on $w_{n}$ and when $w_{n}$ is varied, 
$C_{j}(\mathcal{I}_{n})$ is also varied. However, with a fixed value of 
$w_{n}$, $C_{j}(\mathcal{I}_{n})$ is fixed and 
we are interested only in $t \geq C_{j}(\mathcal{I}_{n})$. If $w_{n} = 0$ then 
the claim follows the induction hypothesis (the instance becomes the one 
with $(n-1)$ jobs).  Assume that the claim holds for some value $w_{n}$.
Now increase an arbitrarily small amount of $w_{n}$ and 
consider a time $t \geq C_{j}(\mathcal{I}_{n})$ (corresponding to the current value of $w_{n}$).
Due to that increase, during period $[t,t+dt]$ the total fractional weight of pending jobs 
varies by $dV(t)$. During the same period $[t,t+dt]$,
the total processing volume done by algorithm is at least
$V(t)^{1/\alpha}dt$ since the machine speed is at least $W(t)^{1/\alpha}$.
Moreover, jobs processed during $[t,t+dt]$ have density at most $\delta_{j}$.
Therefore, $dV(t) \leq  \delta_{j} V(t)^{1/\alpha}dt$.    
So 
$$
\frac{\alpha}{\alpha-1}dW(t)^{(\alpha-1)/\alpha} 
	= \frac{dW(t)}{ W(t)^{1/\alpha}} \leq \delta_{j}dt
$$
This inequality means that in the lemma inequality (\ref{eq:energy-general-main}), 
the decrease in the left-hand side 
is larger than that in the right-hand side while varying the weight of job $n$. 
Hence, the inequality holds for $t \geq C_{j}(\mathcal{I}_{n})$.


Now we consider instance $\mathcal{I}_{n}$ with fixed parameters for job $n$.
We will prove the lemma for $t < C_{j}(\mathcal{I}_{n})$.
Denote $t_{0} = C_{j}(\mathcal{I}_{n})$. Again, rename jobs 
in non-increasing order of their  densities at time $r_{n}$.
(After $r_{n}$, no new job is released and the relative order of jobs is unchanged.)
Let $W_{a}$ be the total weight of pending jobs at $r_{n}$
which have  density smaller than $\delta_{a}$. 
Recall that the total weight of pending jobs at time $t$ is $W(t)$.

Let $k$ be the pending job with the smallest index at time $t$ 
in the instance $\mathcal{I}_{n}$. 
During $[t,t_{0}]$, the machine processes (a part) of job $k$, jobs $k+1,\ldots, j$.
The jobs have density at least $\delta_{j}$.
We deduce
\begin{align*}
\lambda_{j} - \delta_{j} (t-r_{j}) &- \frac{P(s^{c})}{s^{c}}
	= \lambda_{j} - \delta_{j} (t_{0}-r_{j}) -  \frac{P(s^{c})}{s^{c}}
			 + \delta_{j} (t_{0}-t) \\
&\leq \frac{\alpha}{\alpha-1} W(t_{0})^{\frac{\alpha-1}{\alpha}} + \delta_{j} (t_{0}-t) \\
&\leq \frac{\alpha}{\alpha-1} W(t_{0})^{\frac{\alpha-1}{\alpha}}
	+ \delta_{j} \biggl(  \frac{q_{k}(t)}{ W_{k}^{1/\alpha}} + 
		\frac{q_{k+1}(r_{n})}{ W_{k+1}^{1/\alpha}} + \ldots + 
		\frac{q_{j}(r_{n})}{ W_{j}^{1/\alpha}} \biggl)\\
&\leq \frac{\alpha}{\alpha-1} W(t_{0})^{\frac{\alpha-1}{\alpha}}
	 + \int_{W_{j+1}}^{W_{k}} \frac{dz}{z^{1/\alpha}} \\
&\leq \frac{\alpha}{\alpha-1} W_{k}^{\frac{\alpha-1}{\alpha}} 
	= \frac{\alpha}{\alpha-1} W(t)^{\frac{\alpha-1}{\alpha}}	
\end{align*}
The first inequality follows the previous claim, stating that inequality (\ref{eq:energy-general-main}) 
holds for $t \geq t_{0}$. The second inequality follows the fact that at any time the speed 
of the machine is at least $W(t)^{1/\alpha}$. The third inequality
holds since $\delta_{k} \leq \delta_{k+1} \leq \ldots \leq \delta_{j}$
and function $z^{-1/\alpha}$ is decreasing.
The last inequality holds since $W(t_{0}) = W_{j+1}$ and $W_{k} = W(t)$.
\end{proof}

\begin{theorem}
The algorithm has competive ratio at most $\max\{64, 32\alpha/\ln \alpha\}$.
\end{theorem}
\begin{proof}
Recall that the dual has value at least $\min L(x,s,C, F, \lambda)$ where the minimum is taken 
over $(x,s,C,F)$ feasible solution of the primal. The goal is to bound 
the Lagrangian function. 
\begin{align}		\label{eq:lagrangian}
L& (x,s,C, F, \lambda) = \sum_{j} \lambda_{j}p_{j}  + A \int_{0}^{\infty} |F'(t)|dt 
			+ \sum_{j} \int_{r_{j}}^{C_{j}}  \delta_{j}(C_{j} - r_{j}) s_{j}(t)F(t)dt \notag \\
		&- \sum_{i,j} \int_{r_{j}}^{C_{j}} s_{j}(t)F(t)
			\biggl( \lambda_{j} - 2 \frac{P(s(t))}{s(t)} 
								- \delta_{j}(C_{j} - r_{j}) \biggl) \one_{\{s(t) > 0\}}dt	
\end{align}
 Define $L_{1}(x,s,C, F, \lambda)$ as 
$$
\sum_{j} \int_{r_{j}}^{C_{j}} s_{j}(t)F(t)
			\biggl( \lambda_{j} - 2 \frac{P(s(t))}{s(t)}
								- \delta_{j}(C_{j} - r_{j}) \biggl)  \one_{\{s(t) > 0\}} dt 
$$

\begin{claim}		\label{claim:general-energy-L1}
Let $(x,s,C,F)$ be an arbitrary feasible solution of the primal. Then,
\begin{align*}
L_{1}(x,s,C, F, \lambda) 
		\leq \frac{\alpha-1}{(\alpha-1)^{\frac{\alpha}{\alpha-1}}} 
			\mathcal{F}^{*}
			- \frac{1}{2} \int_{0}^{\infty} g\one_{\{F(t) > 0\}} dt 
			- \frac{1}{2} \int_{0}^{\infty} g\one_{\{s^{*}(t) > 0\}}  dt \\
\end{align*}
\end{claim}

\begin{claim}		\label{claim:general-energy-L2}
Let $(x,s,C,F)$ be an arbitrary feasible solution of the primal. 
Define
\begin{align*}
L_{2}(F) := &\sum_{j} \int_{r_{j}}^{C_{j}}  \delta_{j}(C_{j} - r_{j}) s_{j}(t)F(t)dt + A \int_{0}^{\infty} |F'(t)|dt \\
		  & \qquad + \frac{1}{2} \int_{0}^{\infty} g\one_{\{F(t) > 0\}} dt 
			+ \frac{1}{2} \int_{0}^{\infty} g\one_{\{s^{*}(t) > 0\}} dt
\end{align*}
Then, $L_{2}(F) \geq \mathcal{E}^{*}_{2}/4$.
\end{claim}

We first show how to prove the theorem assuming the claims.
By (\ref{eq:lagrangian}), we have
\begin{align*}
L(x,s,C, F, \lambda) 
	&\geq \sum_{j} \lambda_{j}p_{j}  + A \int_{0}^{\infty} |F'(t)|dt
			- \sum\frac{\alpha-1}{(\alpha-1)^{\frac{\alpha}{\alpha-1}}} \mathcal{F}^{*} \\
			& \qquad + \frac{1}{2} \int_{0}^{\infty} g \one_{\{F(t) > 0\}} dt 
			+ \frac{1}{2} \int_{0}^{\infty} g\one_{\{s^{*}(t) > 0\}} dt \\
	&\geq \sum_{j}\lambda_{j}p_{j}  
			- \frac{\alpha-1}{(\alpha-1)^{\frac{\alpha}{\alpha-1}}} \mathcal{F}^{*}
			+ \frac{1}{4}\mathcal{E}^{*}_{2} + \frac{1}{4}\mathcal{E}^{*}_{3}   \\
	&\geq \biggl( 1 - \frac{1}{(\alpha-1)^{1/(\alpha-1)}} \biggl) 
		\biggl( \frac{7}{8}\mathcal{E}^{*}_{1} + 
			 \frac{1}{16}\mathcal{F}^{*} - \frac{3}{16} \mathcal{E}^{*}_{2}
			 \biggl) 
		+ \frac{1}{4}\mathcal{E}^{*}_{2} + \frac{1}{4}\mathcal{E}^{*}_{3}\\
	&=  \biggl( 1 - \frac{1}{(\alpha-1)^{1/(\alpha-1)}} \biggl) 
		\biggl( \frac{7}{8}\mathcal{E}^{*}_{1} + \frac{1}{16}\mathcal{F}^{*} \biggl) 
		+ \biggl( \frac{1}{4} - \frac{3}{16} \biggl) \mathcal{E}^{*}_{2} + \frac{1}{4}\mathcal{E}^{*}_{3}\\
	&\geq  \frac{\ln (\alpha-1)}{(\alpha-1)^{\alpha/(\alpha-1)}} 
		\biggl( \frac{3}{4}\mathcal{E}^{*}_{1} + \frac{1}{8}\mathcal{F}^{*} \biggl) 
		+ \frac{1}{16} \mathcal{E}^{*}_{2} + \frac{1}{4}\mathcal{E}^{*}_{3}		
\end{align*}
where the first and second inequalities are due to Claim~\ref{claim:general-energy-L1} and
Claim~\ref{claim:general-energy-L2}, respectively. The third inequality follows
Corollary~\ref{cor:general-energy-lambda} and $\sum_{j} \lambda_{j}p_{j} \geq \mathcal{F}^{*}$. 
The last inequality is due to the fact that
$(\alpha-1)^{\frac{1}{\alpha-1}} \geq 1 + \frac{\ln (\alpha-1)}{\alpha-1}$
for every $\alpha > 1$.

Besides, the primal objective is at most
$
2(\mathcal{F}^{*} + \mathcal{E}^{*}_{1} + \mathcal{E}^{*}_{2} + \mathcal{E}^{*}_{3})
$. 
Hence, the competitive ratio is at most $\max\{32\alpha/\ln \alpha,64\}$.

In the rest, we prove the claims.

\setcounter{claim}{1}
\begin{claim}		\label{claim:general-energy-L1}
Let $(x,s,C,F)$ be an arbitrary feasible solution of the primal. Then,
\begin{align*}
L_{1}(x,s,C, F, \lambda) 
		\leq \frac{\alpha-1}{(\alpha-1)^{\frac{\alpha}{\alpha-1}}} 
			\mathcal{F}^{*}
			- \frac{1}{2} \int_{0}^{\infty} g\one_{\{F(t) > 0\}} dt 
			- \frac{1}{2} \int_{0}^{\infty} g\one_{\{s^{*}(t) > 0\}}  dt \\
\end{align*}
\end{claim}
\begin{claimproof}
We have
\begin{align*}
L_{1}(x,s,C, F, \lambda) & :=  \sum_{j} \int_{r_{j}}^{C_{j}} s_{j}(t)F(t)
			\biggl( \lambda_{j} - 2\frac{P(s(t))}{s(t)} 
								- \delta_{j}(C_{j} - r_{j}) \biggl)  \one_{\{s(t) > 0\}} dt  \\
		&\leq \sum_{j} \int_{r_{j}}^{C_{j}} s_{j}(t)F(t)
			\biggl( \lambda_{j} - 2\frac{P(s(t))}{s(t)} 
								- \delta_{j}(t - r_{j}) \biggl)  \one_{\{s(t) > 0\}} dt 
\end{align*}
where the inequality holds because the integral for each job $j$ is taken over $r_{j} \leq t \leq C_{j}$. 

Let $T$ be the set of time $t$ such that 
$\frac{\alpha}{(\alpha-1)} W(t)^{\frac{\alpha-1}{\alpha}} \leq \frac{P(s^{c})}{ s^{c}}$.
Then by Lemma~\ref{lem:energy-general-main}, for any $t \in T$
$$
\lambda_{j} - \delta_{j} (t-r_{j})
\leq 2 \frac{P(s^{c})}{ s^{c}}
$$
Therefore,
\begin{align}	\label{eq:L1-null}
\sum_{j} \int_{r_{j}}^{C_{j}} s_{j}(t)F(t)
	\biggl( \lambda_{j} - 2\frac{P(s(t))}{s(t)} 
								- \delta_{j}(C_{j} - r_{j}) \biggl) \one_{\{t \in T\}}dt \leq 0	
\end{align}
since $s^{c} = \arg \min_{z \geq 0} P(z)/z$. Hence, 
\begin{align*}
&L_{1}(x,s,C, F, \lambda) \\
		&\leq \sum_{j} \int_{r_{j}}^{C_{j}} s_{j}(t)F(t)
			\biggl( \lambda_{j} - 2\frac{P(s(t))}{s(t)} 
						- \delta_{j}(t - r_{j}) \biggl)  \one_{\{s(t) > 0\}} \one_{\{F(t) > 0\}} \one_{\{t \notin T\}} dt \\
		&\leq \sum_{j} \int_{r_{j}}^{C_{j}} s_{j}(t)F(t)
			\biggl( \frac{\alpha}{(\alpha-1)} W(t)^{\frac{\alpha-1}{\alpha}} - \frac{P(s(t))}{s(t)} \biggl)  
				\one_{\{s(t) > 0\}}  \one_{\{F(t) > 0\}}  \one_{\{t \notin T\}} dt \\	
		&= \int_{0}^{\infty} s(t)
			\biggl( \frac{\alpha}{(\alpha-1)} W(t)^{\frac{\alpha-1}{\alpha}} - \frac{P(s(t))}{s(t)} \biggl)  
				\one_{\{s(t) > 0\}}  \one_{\{F(t) > 0\}}  \one_{\{t \notin T\}} dt \\
		&\leq \int_{0}^{\infty} 
			\biggl( \frac{\alpha}{(\alpha-1)} W(t)^{\frac{\alpha-1}{\alpha}} \bar{s}(t) - P(\bar{s}(t)) \biggl) 
				\one_{\{s(t) > 0\}}  \one_{\{F(t) > 0\}}  \one_{\{t \notin T\}} dt \\
		&\leq \int_{0}^{\infty}
			\biggl( \frac{\alpha}{(\alpha-1)} W(t)^{\frac{\alpha-1}{\alpha}} \bar{s}(t) - P(\bar{s}(t)) \biggl)   
				 \one_{\{F(t) > 0\}}   \one_{\{s^{*}(t) > 0\}} dt \\							
		&\leq \frac{1}{2}  \int_{0}^{\infty} 
			\biggl( \frac{\alpha}{(\alpha-1)} W(t)^{\frac{\alpha-1}{\alpha}} \bar{s}(t) 
					- P(\bar{s}(t) \biggl)  \one_{\{F(t) > 0\}} dt \\
			& \qquad \qquad + \frac{1}{2}  \int_{0}^{\infty} 
			\biggl( \frac{\alpha}{(\alpha-1)} W(t)^{\frac{\alpha-1}{\alpha}} \bar{s}(t) 
					- P(\bar{s}(t) \biggl)  \one_{\{s^{*}(t) > 0\}} dt
\end{align*}
The first inequality is due to (\ref{eq:L1-null}) and note that if $F(t) = 0$ then the contribution of the term 
inside the integral is 0. The second inequality follows Lemma~\ref{lem:energy-general-main} and 
recall that  $s^{c} = \arg \min_{z \geq 0} P(z)/z$. 
The equality is because $\sum_{j} s_{j}(t)F(t) = \sum_{j} s_{j}(t) = s(t)$
for $t$ such that $F(t) > 0$ (meaning $F(t) = 1$).
The third inequality is due to the first order derivative and $\bar{s}(t)$ is the solution of  
$P'(z) = \frac{\alpha}{(\alpha-1)} W(t)^{\frac{\alpha-1}{\alpha}}$. 
The fourth inequality holds since the term inside the integral is non-negative
and $\{t: t \notin T\} \subset \{t: s^{*}(t) > s^{c}\} \subset \{t: s^{*}(t) > 0\}$.

Replacing $\bar{s}(t) =  (\alpha-1)^{-1/(\alpha-1)} W(t)^{1/\alpha}$
(solution of $P'(z) = \frac{\alpha}{(\alpha-1)} W(t)^{\frac{\alpha-1}{\alpha}}$), 
we get:
\begin{align*}
L_{1}(x,s,C, F, \lambda) 
		&\leq \frac{\alpha-1}{(\alpha-1)^{\frac{\alpha}{\alpha-1}}} 
			 \int_{0}^{\infty} W(t) 
			- \frac{1}{2} \int_{0}^{\infty} g\one_{\{F(t) > 0\}} dt 
			- \frac{1}{2} \int_{0}^{\infty} g\one_{\{s^{*}(t) > 0\}}  dt \\
		&= \frac{\alpha-1}{(\alpha-1)^{\frac{\alpha}{\alpha-1}}} 
			\mathcal{F}^{*}
			- \frac{1}{2} \int_{0}^{\infty} g\one_{\{F(t) > 0\}} dt 
			- \frac{1}{2} \int_{0}^{\infty} g \one_{\{s^{*}(t) > 0\}}  dt
\end{align*}
\end{claimproof}

\begin{claim}		\label{claim:general-energy-L2}
Let $(x,s,C,F)$ be an arbitrary feasible solution of the primal. 
Define
\begin{align*}
L_{2}(F) := &\sum_{j} \int_{r_{j}}^{C_{j}}  \delta_{j}(C_{j} - r_{j}) s_{j}(t)F(t)dt + A \int_{0}^{\infty} |F'(t)|dt \\
		  & \qquad + \frac{1}{2} \int_{0}^{\infty} g\one_{\{F(t) > 0\}} dt 
			+ \frac{1}{2} \int_{0}^{\infty} g\one_{\{s^{*}(t) > 0\}} dt
\end{align*}
Then, $L_{2}(F) \geq (\mathcal{E}^{*}_{2} + \mathcal{E}^{*}_{3})/4$.
\end{claim}
\begin{claimproof}
Consider the algorithm schedule. An \emph{end-time} $u$ is a moment in the algorithm schedule such
that the machine switches from the idle state to the sleep state. 
Conventionally, the first end-time in the schedule is 0. 
Partition the time line into phases. A \emph{phase} $[u,v)$ is a time interval such that $u,v$
are consecutive end-times. Observe that in a phase, the schedule
has transition cost $A$ and some new job is released in a phase
(otherwise the machine would not switch to non-sleep state).
We will prove the claim on every phase. In the following, 
we are only interested in phase $[u,v)$ and define 
\begin{align*}
\restr{L_{2}(F)}{[u,v)} := &\sum_{j: u \leq r_{j} < v} \int_{r_{j}}^{C_{j}}  \delta_{j}(C_{j} - r_{j}) s_{j}(t)F(t)dt 
					+ A \int_{0}^{\infty} |F'(t)|dt \\
		  & \qquad + \frac{1}{2} \int_{u}^{v} g\one_{\{F(t) > 0\}} dt 
			+ \frac{1}{2} \int_{u}^{v} g\one_{\{s^{*}(t) > 0\}} dt
\end{align*}

By the algorithm, the static energy on machine $i$ during 
its idle time is $A$, i.e., $\int_{u}^{v}  g\one_{\{s^{*}(t) = 0\}} dt = A$. 
If during $[u,v)$, the schedule induced by solution $(x,s,C,F)$ makes a transition from
non-sleep state to sleep state or inversely then 
$\restr{L_{2}(F)}{[u,v)} \geq \frac{1}{2} \int_{u}^{v} g \one_{\{s^{*}(t) > 0\}} dt + A$. Hence
$$
\restr{L_{2}(F)}{[u,v)} \geq \frac{1}{2} \biggl( \int_{u}^{v} g \one_{\{s^{*}(t) > 0\}} dt 
	      	+  \int_{u}^{v}  g \one_{\{s^{*}(t) = 0\}}dt + A \biggl) 
		= \frac{1}{2}\restr{\mathcal{E}^{*}_{2}}{[u,v)} + \frac{1}{2}\restr{\mathcal{E}^{*}_{3}}{[u,v)}.
$$

If during $[u,v)$, the schedule induced by solution $(x,s,C,F)$ makes no transition (from
non-sleep static to sleep state or inversely) then either $F(t) = 1$ or $F(t) = 0$ for every 
$t \in [u,v]$. Note that by definition of phases, some job is released during $[u,v)$.

\paragraph{Case 1: $F(t) = 1 ~\forall u \leq t \leq v$.} 
Hence,
\begin{align*}
\restr{L_{2}(F)}{[u,v)} &\geq \frac{1}{2} \int_{u}^{v} g \one_{\{s^{*}(t) > 0\}} dt 
		+ \frac{1}{2} \int_{u}^{v} g \one_{\{F(t) =1\}} dt 
		=  \frac{1}{2} \int_{u}^{v} g \one_{\{s^{*}(t) > 0\}} dt 
		+ \frac{1}{2} \int_{u}^{v} g dt \\
	      & \geq \frac{1}{2} \int_{u}^{v} g \one_{\{s^{*}(t) > 0\}} dt 
	      	+ \frac{1}{4} \int_{u}^{v}  g \one_{\{s^{*}(t) = 0\}}dt + \frac{A}{4} \\
	      & \geq \frac{1}{4} \biggl( \int_{u}^{v} g \one_{\{s^{*}(t) > 0\}} dt 
	      	+  \int_{u}^{v}  g \one_{\{s^{*}(t) = 0\}}dt + A \biggl) 
		= \frac{1}{4}\restr{\mathcal{E}^{*}_{2}}{[u,v)} + \frac{1}{4}\restr{\mathcal{E}^{*}_{3}}{[u,v)}
\end{align*}
where the second inequality follows since the total idle duration in $[u,v)$ incurs a cost $A$
(so the machine switches to sleep state at time $v$).

\paragraph{Case 2: $F(t) = 0 ~\forall u \leq t \leq v$.}  As the machine is in the sleep state 
during $[u,v)$ in solution $(x,s,C,F)$, all jobs 
released in $[u,v)$ are completed later than $v$. 
By the algorithm, the total weighted flow-time of such jobs is at least the static energy of the algorithm
during $[u,v)$. In other words,
\begin{align*}
\restr{L_{2}(F)}{[u,v)} &\geq \sum_{j: u \leq r_{j} < v} \int_{r_{j}}^{C_{j}}  \delta_{j}(C_{j} - r_{j}) s_{j}(t)F(t)dt  
		+ \frac{1}{2} \int_{u}^{v} g \one_{\{s^{*}(t) > 0\}} dt \\
		& \geq \int_{u}^{v} g \one_{\{s^{*}(t) = 0\}} dt 
	      	+ \frac{1}{2} \int_{u}^{v}  g \one_{\{s^{*}(t) > 0\}}dt  \\
		& \geq \frac{1}{2} \int_{u}^{v} g \one_{\{s^{*}(t) = 0\}} dt 
	      	+ \frac{1}{2} \int_{u}^{v}  g \one_{\{s^{*}(t) > 0\}}dt  + \frac{A}{2}
	        = \frac{1}{2} \restr{\mathcal{E}^{*}_{2}}{[u,v)} +  \frac{1}{2} \restr{\mathcal{E}^{*}_{3}}{[u,v)}
\end{align*}
where the third inequality is again due to the fact that 
the total idle duration in $[u,v)$ incurs a static energy $A$. 
\end{claimproof}

The above proofs of the claims complete the theorem proof.
\end{proof}

\bibliographystyle{plainnat}
\bibliography{scheduling}
\newpage

\appendix

\section*{Appendix}

\setcounter{lemma}{-1}
\begin{lemma}[Weak duality]
Consider a possibly non-convex optimization problem
\begin{align*}
p^{*} := \min_{x} f_{0}(x) ~: \quad f_{i}(x) \leq 0, \quad i=1,\ldots,m. 
\end{align*}
where $f_{i}: \mathbb{R}^{n} \rightarrow \mathbb{R}$ for $0 \leq i \leq m$.
Let $\mathcal{X}$ be the feasible set of $x$. 
Let $L: \mathbb{R}^{n} \times \mathbb{R}^{m} \rightarrow \mathbb{R}$ be the Lagragian function
$$
L(x,\lambda) = f_{0}(x) + \sum_{i=1}^{m} \lambda_{i} f_{i}(x).
$$
Define $d^{*} = \max_{\lambda \geq 0} \min_{x \in \mathcal{X}} L(x,\lambda)$
where $\lambda \geq 0$ means $\lambda \in \mathbb{R}^{m}_{+}$. 
Then $p^{*} \geq d^{*}$.
\end{lemma}
\begin{proof}
We observe that, for every feasible $x \in \mathcal{X}$, and every $\lambda \geq 0$,
$f_{0}(x)$ is bounded below by $L(x,\lambda)$:
$$
\forall x \in \mathcal{X}, ~\forall \lambda \geq 0: ~ f_{0}(x) \geq L(x,\lambda)
$$

Define a function $g: \mathbb{R}^{m} \rightarrow \mathbb{R}$ such that
$$
g(\lambda) := \min_{z} L(z,\lambda) = \min_{z} f_{0}(z) + \sum_{i=1}^{m} \lambda_{i} f_{i}(z)
$$
As $g$ is defined as a point-wise minimum, it is a concave function.  

We have, for any $x$ and $\lambda$, $L(x,\lambda) \geq g(\lambda)$. Combining with 
the previous inequality, we get 
$$
\forall x \in \mathcal{X}: ~ f_{0}(x) \geq g(\lambda)
$$
Taking the minimum over $x$, we obtain
$
\forall \lambda \geq 0: ~ p^{*} \geq g(\lambda).
$
Therefore,
$$
p^{*} \geq \max_{\lambda \geq 0} g(\lambda) = d^{*}.
$$
\end{proof}

\end{document}